%% file: PostSelectionInference.tex
\documentclass[11pt,twoside]{article}

\input{Packages}

\input{MyPageStyle}

\input{NewCommands}

\newcommand{\titleshort}{Post-selection inference by estimator augmentation}
\newcommand{\authorshort}{Min and Zhou}

\newcommand{\proglang}{\textsf}

\newcommand{\compresslist}{
 \setlength{\itemsep}{5pt}
 \setlength{\parskip}{0pt}
 \setlength{\parsep}{0pt}
}

\usepackage[colorlinks=true,citecolor=blue,linkcolor=blue]{hyperref}
\usepackage{multirow}

\numberwithin{equation}{section}
\newtheorem{theorem}{Theorem}

\newtheorem{proposition}[theorem]{Proposition}
\newtheorem{corollary}[theorem]{Corollary}

\theoremstyle{definition}

\newtheorem{remark}{Remark}

\newtheorem{routine}{Algorithm}

\begin{document}

\title{Constructing Confidence Sets After Lasso Selection by Randomized Estimator Augmentation}
\author{Seunghyun Min$^*$ and Qing Zhou\thanks{UCLA Department of Statistics. Email: seunghyun@ucla.edu; zhou@stat.ucla.edu}}
\date{}
\maketitle

\begin{abstract}
Although a few methods have been developed recently for building confidence intervals after model selection, how to construct confidence sets for joint post-selection inference is still an open question. In this paper,
we develop a new method to construct confidence sets after lasso variable selection, with strong numerical support for its accuracy and effectiveness. A key component of our method is to sample from the conditional distribution of the response $y$ given the lasso active set, which, in general, is very challenging due to the tiny probability of the conditioning event. We overcome this technical difficulty by using estimator augmentation to simulate from this conditional distribution via Markov chain Monte Carlo given any estimate $\tdmu$ of the mean $\mu_0$ of $y$. We then incorporate a randomization step for the estimate $\tdmu$ in our sampling procedure, which may be interpreted as simulating from a posterior predictive distribution by averaging over the uncertainty in $\mu_0$. Our Monte Carlo samples offer great flexibility in the construction of confidence sets for multiple parameters. 
Extensive numerical results show that our method is able to construct confidence sets with the desired coverage rate and, moreover, that the diameter and volume of our confidence sets are substantially smaller in comparison with a state-of-the-art method.

\end{abstract}

\section{Introduction}
\label{sec:intro}

Assuming that a random vector $y \in \R^n$ follows a multivariate Gaussian distribution,
\begin{align}\label{eq:ydistn}
y = \mu_0 + \varepsilon, \qquad \varepsilon \sim \calN_n(0, \sigma^2 \bfI_n),
\end{align}
we wish to make inference on the unknown mean vector $\mu_0\in\R^n$ after observing $y$. Given a set of $p$ covariates $X = \big[X_1 | \cdots | X_p\big] \in \R^{n \times p}$, a common approach is to approximate $\mu_0$ with a linear combination $X\beta$ for $\beta\in\R^p$. 
When the number of covariates is large, we often face the situation that only some of them can be included in the linear approximation. They may be selected manually or by a certain model selection method. Let $A \subset \{1, \cdots p\}$ be the set of selected covariates. After the selection step is done, our goal shifts to constructing a linear model that can best approximate $\mu_0$ with only the selected covariates $X_A=(X_j,j\in A)$.
Then the parameter of interest 
\begin{equation}\label{eq:nu}
\nu := X_A^+\mu_0 = \argmin_{\beta\in\R^{|A|}} \| \mu_0 - X_A \beta \|^2_2
\end{equation}
is defined by the projection of $\mu_0$ onto the space spanned by $X_A$. Inference on $\nu$ is called post-selection inference \citep{pots:91}. In large-scale analysis, model selection is usually applied as an initial screening for important variables or features. In these applications, naive methods based on the standard $t$-statistic or interval will not provide valid inference for the selected variables due to selection bias in the screening step \citep{tibs:etal:16,liu:etal:18}. By conditioning on the model selection event, post-selection inference provides reliable quantification of the significance of a selected variable, which is critical for follow-up investigations. Another appealing feature for post-selection inference is that it is valid without assuming a true linear model for $y$, only regarding the selected model as an approximation for $\mu_0$ \citep{berk:etal:13}, which greatly relaxes its model assumptions. 

When the selection step is done independently from $y$, for example by using another independent dataset or by pre-given information, inference on $\nu$ can be easily done with conventional methods. The distribution of the least-squares estimator $\hat \nu = X_A^+ y$ simply follows a Gaussian distribution. However, the problem becomes much more challenging when the selection step is data-driven and uses the same $y$. In such a case, conditioning on the selected active set $A$, the sampling distribution of $y$ is restricted to a potentially irregular subset of $\R^n$. This problem is further complicated for high-dimensional data with $p>n$. Several lines of recent work have laid down the theoretical foundations and developed novel methods for post-selection inference on high-dimensional data. \cite{tibs:etal:16} develop a truncated Gaussian statistic to test the significance of an entering variable in each step of a sequential regression method, which generalizes the earlier work by \cite{lock:etal:14}. \cite{tibs:etal:18} establish uniform convergence  properties of this statistic, without normal assumption, as $n\to\infty$ and $p$ stays fixed. \cite{lee:etal:16} build exact confidence intervals for individual components $\nu_j$ of $\nu$ in \eqref{eq:nu}, where the set $A$ is the support of the lasso \citep{tibs:96}, which we will call Lee's method hereafter. The authors show that conditioning on the active set of the lasso is equivalent to imposing polyhedral constraints on $y$, a key idea used in \cite{tibs:etal:16} as well. \cite{tian:tayl:17} have established asymptotic results for Lee's method without imposing Gaussian assumption. 
\cite{tayl:tibs:18} further generalize Lee's method to generalized linear models, Cox's proportional hazards model and Gaussian graphical models. By conditioning on a smaller and more robust subset of the lasso active set, \cite{liu:etal:18} develop a more efficient method that produces shorter intervals. The randomize inference in \cite{tiantaylor18} improves the numerical stability of Lee's method with substantial power gain. \cite{Bachoc20} suggest general methods to construct asymptotically valid confidence intervals post model selection.

In this article, we seek to make inference on $\nu$ \eqref{eq:nu} with the model selected by the lasso. That is, the set $A$ is the support of 
\begin{align}\label{eq:lassodef}
\hbeta(y):=\argmin_{\beta\in\R^p}\frac{1}{2n}\|y-X\beta\|_2^2 + \lambda\sum_{i=1}^pw_i|\beta_i|,
\end{align}
where $w_i>0$ and are set to 1 by default. However, in contrast to \cite{lee:etal:16} and the other methods reviewed above, we aim at constructing not only confidence intervals for an individual $\nu_j$, but also confidence sets for $\nu_B$, where $B$ contains an arbitrary subset of $A$. To the best of our knowledge, methods for constructing confidence sets after model selection have not been proposed in the literature. Although one might consider simultaneously covering all $\nu_j$, $j\in B$ by controlling family-wise error rate, such an approach would be very stringent for large $B$, as verified numerically in comparison to our proposed method. On the other hand, the method of \cite{lee:etal:16} critically relies on the cumulative distribution function of a \emph{univariate} Gaussian distribution truncated to the union of $2^{|A|}$ intervals. It seems highly intractable to generalize their technique for joint inference on a potentially large set of $\nu_j$. 
Moreover, although Lee's method preserves the coverage rate at a desirable level, their confidence intervals are not always informative. In particular, their method sometimes produces infinite intervals with $\infty$ or $-\infty$ as the upper or lower bound, severely limiting its practical applications. \cite{kiva:leeb:18} show that the expected interval length of Lee's method can be infinity under certain condition, which is frequently satisfied in their simulation study.

A key difference between our method and the existing ones is that ours is built upon sampling of $y^*$ that leads to the same active set of the lasso, i.e. $[y^*\mid \supp(\hbeta(y^*))=A]$, where $A=\supp(\hbeta(y))$ is computed from the observed data $(X,y)$.  This sampling-based approach allows for the construction of confidence sets for joint inference on any subset of the parameter vector $\nu$. It also offers great flexibility in choosing the statistic for inference, as the distribution of  any function $T(y^*)$ can be readily approximated from a large sample of $y^*$. However, this conditional sampling is a challenging computational problem, since the event $\{\supp(\hbeta(y^*))=A\}$ is in general a rare event, especially when $p$ is large. To complete this difficult task, we develop a novel conditional sampler via the method of estimator augmentation \citep{zhou:14}, given a point estimate $\tdmu$ of $\mu_0$.
To protect our method from a poor estimate $\tdmu$, we introduce a randomization step to draw a uniform sample of $\tdmu$ from a set $\wh C$, which in conjunction with our conditional sampling of $y^*$ produces an efficient and accurate tool for joint inference after lasso selection. The set $\wh C$ can be seen as a way to incorporate the uncertainty in estimating $\mu_0$ from $y$, prior to or unconditional on model selection, which allows for an adaptive and robust approximation of the distribution $[y^*\mid \supp(\hbeta(y^*))=A]$. When used for inference on individual parameters, our method often builds much shorter confidence intervals than Lee's method, while achieving a comparable coverage rate. Furthermore, our method, by design, does not produce infinite intervals or sets. Our post-selection inference method has been implemented in the \proglang{R} package \textbf{EAinference}, which includes many other applications of estimator augmentation and related simulation-based inference tools.

The rest of the paper is organized as follows.
In Section~\ref{sec:psi}, we introduce the key ingredients of our method: how to build confidence sets via conditional sampling and how to implement the randomization step. Section~\ref{sec:condsamp} develops a Markov chain Monte Carlo (MCMC) algorithm for the conditional sampling. Section~\ref{sec:numrslt} demonstrates empirically the effectiveness and accuracy of the confidence sets constructed by our method, including comparisons with Lee's method. We conclude the paper with some remarks and discussion in Section~\ref{sec:conc}. Proofs of technical results are provided in Section~\ref{sec:proof}. 

Notation used throughout the paper is defined here.
Let $\N_k$ denote the set $\{1, \cdots, k\}$. Let ${\bf 1}_{[k]}$ be a $k$-vector of ones.
Denote by $Z_i$ the $i$-th column or the $i$-th component of $Z$ when $Z$ is a matrix or a vector, respectively. 
Correspondingly, we define $Z_A:=(Z_i)_{i \in A}$ and $Z_{-i}:=(Z_j)_{j\ne i}$. For a matrix $Z$, let $Z_{AB}$ be the submatrix consisting rows in $A$ and columns in $B$.
The superscript $+$ is used for Moore-Penrose inverse. 
Denote by $\row(X)$ and null$(X)$ the row space and the null space of a matrix $X$, respectively.

\section{Post-selection inference}\label{sec:psi}
\subsection{Basic idea}\label{sec:basic}
For the lasso estimate $\hbeta(y)$, let $\calA(y)=\text{supp}(\hbeta(y))$ be the set of active variables.
Given the active set $\calA(y) = A$,
the parameter of interest $\nu=X_A^+ \mu_0$ \eqref{eq:nu} is the coefficient vector for the projection of $\mu_0 = \mathbb{E}[y]$ onto span$(X_A)$. 
Our goal is to construct a confidence set $\wh I_B(\alpha)$ such that
\begin{equation}\label{eq:hatI}
\mathbb{P}\Big\{\nu_B \in \wh I_B(\alpha) \Big| \calA(y) = A\Big\} \geq 1-\alpha \text{\quad for } B \subset \N_{|A|},
\end{equation}
where the probability is taken with respect to $y\sim\dnorm_n(\mu_0,\sigma^2\bfI_n)$.
In particular, when $B = \{j\}$, $\wh I_B(\alpha)$ is a confidence interval for $\nu_j$, which we denote by $\wh I_j(\alpha)$. A natural choice for the center of the confidence set is $\hat \nu_B := [X_A^+y]_B$.
This problem is, however, more complicated than it may look. Since we have selected variables using lasso, the distribution of 
$X_A^+ y$ given $\calA(y) = A$ is no longer a Gaussian distribution, as the support of $y$ is now only a proper subset of $\mathbb{R}^n$. 

We will develop a simulation-based approach. Note that the conditioning event $\{\calA(y)=A\}$ restricts
our sampling to those $y$ for which the lasso $\hbeta(y)$ selects exactly the same variables in $A$, which is usually a rare event. Thus, it is almost impossible to use bootstrap to draw from $[y^* \mid \calA(y^*)=A]$,
where $y^*$ denotes a sample drawn from an (estimated) distribution of $y$.
However, estimator augmentation \citep{zhou:14} enables us to simulate from this conditional distribution, with a point estimate $\tdmu$ for $\mu_0$, by an MCMC algorithm; see Section~\ref{sec:MHLS} for the details.

Suppose we have drawn a large sample of $y^*$ by this Monte Carlo algorithm.
One could use
 $[X_A^+(y^{*} - \tdmu)  \mid \calA(y^*)=A]$, which can be easily estimated from the samples of $y^*$,
 to approximate $[X_A^+(y - \mu_0) \mid \calA(y)=A]$ and build a confidence set for $\nu_B$. However, due to the dependency of these distributions on $\tdmu$ and $\mu_0$, the former is in general not a uniformly consistent estimator for the latter; see \cite{leeb:pots:06} for related discussion on estimating the conditional distribution $[\hbeta(y)\mid \calA(y)=A]$. In practice, this means that a poor choice of $\tdmu$ often results in poor coverage. To overcome this difficulty, we develop a robust method which randomizes the plug-in estimate $\tdmu$. As it will become clear, 
our approach is to bound the relevant quantiles of $[X_A^+(y - \mu_0) \mid \calA(y)=A]$ in order to perform conservative inference as stated in \eqref{eq:hatI}.

\subsection{The randomization step}\label{sec:rand}
We will first develop our method for constructing confidence intervals for $\nu_j$, which will be generalized in Section \ref{sec:JointInference} to joint inference on $\nu_B$.
Let $q_{j,\gamma}(\mu)$ be the $\gamma$-quantile of the distribution
\begin{align} \label{eq:CImu}
[\{X_A^+(y^*-\mu)\}_j\mid \calA(y^*)=A], 
\end{align}
where $y^*=\mu+\veps^*$ and $\veps^*\sim\dnorm_n(0,\sigma^2\bfI_n)$. For $\gamma\in(0,1)$, construct an interval
\begin{equation}\label{eq:CI}
\xi_j(\mu,\gamma) := \Big[ \hat \nu_{j}  - q_{j,1-\gamma/2}(\mu), \hat \nu_{j}  - q_{j,\gamma/2}(\mu) \Big]. 
\nonumber
\end{equation}
By definition, the coverage rate of $\xi_j(\mu_0,\gamma)$ is $1-\gamma$. Call $\xi_j(\mu_0, \gamma)$ the oracle interval.  Of course, $\mu_0$ is unknown so we need an estimate $\tdmu$ in place of $\mu_0$ to construct a practical interval $\xi_j(\tdmu, \gamma)$. 
One problem is that the conditional distribution in \eqref{eq:CImu} depends on $\mu$ due to the selection event and  $q_{j,\gamma}(\tdmu)$ is not guaranteed to converge uniformly to $q_{j,\gamma}(\mu_0)$.
To alleviate this issue, we propose a method to randomize the point estimate $\tdmu$, which is motivated by the following conservative construction.

Suppose we have a set $C \subset \R^n$ such that $\mu_0 \in C$. 
For $\gamma<1/2$, define
\begin{align}
q_{j,1-\gamma}^*(C)&=\max_{\mu\in C}q_{j,1-\gamma}(\mu), \label{eq:defupmax}\\
q_{j,\gamma}^*(C)&=\min_{\mu\in C}q_{j,\gamma}(\mu). \label{eq:deflowmax}
\end{align}
Then it follows that the coverage rate
of the interval $[\hat \nu_j-q_{j,1-\gamma/2}^*(C),\hat \nu_j-q_{j,\gamma/2}^*(C)]$ is at least $1-\gamma$.
A possible choice for the set $C$ is a confidence set $\wh C$ for the mean $\mu_0$, unconditional on the selected model. We have the following result about using such an interval for a conservative coverage. Recall $\nu=X_A^+\mu_{0}$ and $\hat{\nu}=X_A^+ y$.

\begin{proposition}\label{lm:converage}
Suppose $y\sim\dnorm_n(\mu_0,\sigma^2\bfI_n)$ and $\wh C$ is a $1-\alpha/2$ confidence set for $\mu_0$, independent of $y$.
Let $\xi_j^*(\wh C) = [\hat \nu_j-q_{j,1-\alp/4}^*(\wh C),\hat \nu_j-q_{j,\alp/4}^*(\wh C)]$.
Then we have
\begin{align} \label{eq:conservativepr}
\Prob\left\{\nu_j\in  \xi_j^*(\wh C) \bigg| \calA(y) = A \right\}\geq 1-\alp. 
\end{align}
\end{proposition}

Proposition~\ref{lm:converage} shows that we can construct a valid confidence interval for post-selection
inference, provided a $(1-\alpha/2)$ confidence set $\wh C$ for $\mu_0$. 
Since its length is determined by the worst scenarios over all $\mu\in\wh C$ as in \eqref{eq:defupmax} and \eqref{eq:deflowmax}, the confidence interval $\xi_j^*(\wh C)$ can be overly conservative, as shown by our numerical results in Section \ref{sec:randsim}. Moreover, the assumption that $\wh C$ is independent of $y$ can be strong unless we use sample-splitting, which does not align well with the purpose of post-selection inference.
However, it provides good intuition for the use of the set $\wh C$ in our proposed randomization step, as described in what follows.

Hereafter, suppose that $\hmu$ is the center of $\wh C=\wh C(y)$ constructed from $y$.
Let $\tdmu$ be uniformly distributed over $\wh C$, i.e. $\tdmu \sim \calU(\wh C)$, and
${q}_{j,\gamma}(\wh C)$ be the $\gamma$-quantile of the distribution
\begin{align} \label{eq:qhatC}
[\{X_A^+(y^{*}-\hmu)\}_j\mid \calA(y^{*})=A],
\end{align}
where $y^{*}=\tdmu+\veps^*$ and $\veps^*$ is independent of $\tdmu$.
Construct an interval $\xi_j(\wh C)$ with $q_{j,\gamma}(\wh C)$ as
\begin{equation} \label{eq:CIfinal}
\xi_j(\wh C) := [\hat \nu_j - {q}_{j,1-\alpha/4}(\wh C), \hat \nu_j - {q}_{j,\alpha/4}(\wh C)].
\end{equation}
Note that the quantile ${q}_{j,\gamma}(\wh C)$ is calculated from a randomized plug-in estimate $\tdmu$ over the confidence set $\wh C$, which takes into account the uncertainty in $\tdmu$. Thus, this interval incorporates more variation than using a fixed point estimate $\hmu$ as in $\xi_j(\hmu,\alpha)$. 

Below, we present our main algorithm for constructing the confidence interval $\xi_j(\wh C)$.
Let $\partial \wh C$ denote the boundary of $\wh C$.

\begin{routine}\label{alg:CI}
Constructing interval $\xi_j(\wh C)$, $j \in \N_{|A|}$: 
\begin{enumerate}
\item Draw $\tdmu^{(k)}$ uniformly from $\partial \wh C$ for $k =1,\ldots, K$. 
\item For each $k$, draw $\{y^*_{ki},i=1,\ldots,N\}$ from $[y^* | \calA(y^*) = A]$ where $y^* \sim \calN_n(\tdmu^{(k)}, \sigma^2 {\bfI_n})$.
\item Estimate $q_{j,\gamma}(\wh C)$ by the quantiles of $\{[X_A^+(y^{*}_{ki}-\hmu)]_j, \forall k,i\}$
and construct $\xi_j(\wh C)$ \eqref{eq:CIfinal} with the estimated quantiles.
\end{enumerate}
\end{routine}
\noindent Here, we draw $\tdmu^{(k)}$ from $\partial \wh C$ for efficiency. Since $\wh C$ is usually an $n$-dimensional ellipsoid, uniform points in $\wh C$ will be close to its boundary when $n$ is large.

Our randomization of the plug-in estimate $\tdmu$ can be interpreted from a Bayesian perspective, regarding $\mu_0$ as a random vector. As discussed above, a confidence interval can be constructed if we have a good approximation to the distribution $[y^* \mid \calA(y^*)=A, \mu_0]$, where $y^*\mid \mu_0 \sim\dnorm_n(\mu_0,\sigma^2\bfI_n)$ is a new response vector independent of $y$. From a Bayesian perspective, a good approximation that takes into account the uncertainty in $\mu_0$ is the posterior predictive distribution 
\begin{align*}
p\big(y^*\mid \calA(y^*)=A,y\big) = \int p\big(y^*\mid \calA(y^*)=A, \mu_0\big) p(\mu_0\mid y)d\mu_0,
\end{align*}
where $p(\mu_0\mid y)$ is a posterior distribution for $\mu_0$.
Regarding $\calU(\wh C)$, the uniform distribution over $\wh C(y)$, as a posterior distribution for $\mu_0$, steps 1 and 2 in Algorithm \ref{alg:CI} can be interpreted as sampling from the above posterior predictive distribution.
Drawing $\tdmu$ in step 1 is equivalent to drawing samples from $p(\mu_0 \mid y)$ and 
drawing $y^*$ in step 2 is equivalent to sampling from $p(y^* \mid \calA(y^*)=A, \tdmu)$, which can be done by our Monte Carlo algorithm to be developed in the next section. In step 3, we find the quantiles of $[X_A^+(y^*-\hmu)\mid \calA(y^*)=A, y]$, where $\hmu$, the center of $\wh C$, is the posterior mean of $\mu_0$.

\subsection{Joint inference}\label{sec:JointInference}
Given the samples of $y^*$ drawn by Algorithm~\ref{alg:CI}, we can easily approximate the conditional distribution $[T(y^*)\mid \calA(y^{*})=A]$ for any function $T(\cdot)$ and carry out many inferential tasks. In particular, we extend our method to the construction of confidence sets for $\nu = X_A^+ \mu_0$. 

Recall $\hat \nu = X_A^+ y$ and let $q=|A|$. Given a matrix $H \in \R^{m \times q}$ for some $m\leq q$,
we wish to make inference on the parameter vector $H\nu\in\R^m$. 
Generalizing \eqref{eq:qhatC}, let $q_{H, \gamma}(\wh C; \ell_\delta)$ be the $\gamma$-quantile of the distribution
\begin{align} \label{eq:DistJointqhatC} 
\Big[\left\| H (X_A^+y^{*}-X_A^+\hmu) \right\|_\delta ~\Big|~ \calA(y^{*})=A\Big], 
\end{align}
where $\delta\in[1,\infty]$ specifies a particular $\ell_\delta$ norm used in our construction. From the above Bayesian interpretation, \eqref{eq:DistJointqhatC} approximates $[\| H\hat \nu - H\nu\|_\delta \mid \calA(y)=A]$ as its posterior predictive estimate. Then we construct a $1-\alpha$ confidence set for $H \nu$ as an $\ell_\delta$ ball
\begin{align} \label{eq:JointqhatC}
\xi_{H}(\wh C; \ell_\delta) := \left\{ \eta \in\R^m : \|\eta- H \hat \nu \|_\delta \leq q_{H, 1-\alpha/2}(\wh C; \ell_\delta) \right\},
\end{align}
where $\wh C$, as in \eqref{eq:CIfinal}, is a $1-\alpha/2$ confidence set for the mean $\mu_0$. For example, one can construct a confidence set for $\nu$ by letting $H = \mathbf{I}_{q}$. If one is interested in constructing a confidence set for the first two components in $A$, we can let $H = [e_1, e_2]^\trans$, where $e_j$ is the $j$-th standard basis vector in $\R^q$. In general, $\xi_{H}(\wh C; \ell_\delta)$ is a confidence set for some linear transformation of $\nu$. 

Now the remaining question is how to build the ($1-\alp/2$) confidence set $\wh C$ for $\mu_0$, unconditional on the selected model. There are a few methods that may be used to construct such a confidence set for high-dimensional regression problems, such as \cite{nick:vand:13, ewald18, zhou:etal:19}. We apply two different methods in this work. 
The first method is a two-step method, consisting of a projection and a shrinkage step \citep{zhou:etal:19}. This method builds an ellipsoid-shaped confidence set with different radii for strong and weak signals. The radius and center for weak signals are constructed using Stein's method. Denote by $\wh C_S$ and $\hmu_S$ the confidence set and its center by this method. It is shown by \cite{zhou:etal:19} that $\wh C_S$ is asymptotically honest,
\begin{align*}
\liminf_{n\to\infty} \inf_{\mu_0\in\R^n} \Prob(\mu_0 \in \wh C_S) \geq 1-\alp/2,
\end{align*}
where $\Prob$ is taken with respect to the distribution of $y$ in \eqref{eq:ydistn}. However, this method replies on sample-splitting in its construction, which adds another level of complexity in the application of our post-selection inference. Thus, we develop a second and simpler method, based on a given subset of covariates $X_A$. Let $A_0=\supp(\beta_0)$ be the true support such that $\mu_0=X_{A_0}\beta_{0A_0}$. If $A_0 \subset A$, then $X_A^+y \sim \calN_{|A|}(\beta_{0A}, \sigma^2(X_A^\trans X_A)^{-1})$. 
From this fact, we build a confidence set $\wh D$ for $\beta_{0A}$ which defines a confidence set $\wh C = X_A \wh D$ for $\mu_0$. The confidence set and its center built this way are denoted by $\wh C_A$ and $\hmu_A$. 
A convenient choice of $A$ would be $\calA(y)$, the support of lasso, and under this choice
$\wh C_A$ will achieve the nominal confidence level if $\Prob(\calA(y)={A})\to 1$ for some ${A}\supset A_0$ as $n\to \infty$. 
We will compare the performance of these two methods in Section \ref{sec:randsim} on simulated data. The comparison suggests that the second method usually achieves comparable coverage as the first method, while being more coherent with our post-selection inference procedure in practice. Therefore, we use the second method by default for all the numerical results in this work.

Now we summarize the steps of our post-selection inference on $H\nu$. 
\begin{routine}\label{alg:idea}
Constructing $\xi_H(\wh C; \ell_\delta)$: 
\begin{enumerate}
\item Construct ($1-\alp/2$) confidence set $\wh C_A$ centering at $\hmu_A$.
\item Apply Algorithm~\ref{alg:CI} with $\wh C=\wh C_A$ and $\hmu=\hmu_A$.
\item Estimate $q_{H, \gamma}(\wh C; \ell_\delta)$ by the quantile of $\{\|HX_A^+(y^{*}_{ki}-\hmu_A)\|_\delta, \forall k,i\}$ and construct $\xi_{H}(\wh C; \ell_\delta)$ in \eqref{eq:JointqhatC} with the estimated quantile.
\end{enumerate}
\end{routine}

\begin{remark}\label{rmk:lbd}
We have implicitly assumed a fixed tuning parameter $\lambda$ in \eqref{eq:lassodef} so far, but we observe that our method works well even using a $\lambda$ chosen in a data-dependent way. This is very appealing in applications: One may simply use lasso, with a data-dependent $\lambda$, to identify potentially importance variables, followed by our inference tool to construct an interval for each. Although not the focus of this paper, for low-dimensional data ($p<n$), the confidence set $\wh C$ for $\mu_0$ can be constructed with $A=\N_p$ by the sampling distribution of the least-squares estimator. 
\end{remark}

\subsection{An illustration}

In Section~\ref{sec:rand}, we covered four different methods for constructing $\wh I_j(\alpha)$. First, the oracle interval $\xi_j(\mu_0,\alpha)$ is constructed assuming the true mean $\mu_0$ is known (the oracle). This is not a practical method and is used for illustration only. 
Second, $\xi_j(\hmu,\alpha)$ uses an estimate $\hmu$ in place of $\mu_0$. Our main proposal $\xi_j(\wh C)$, presented in Algorithm \ref{alg:CI}, randomizes the plug-in estimate of $\mu_0$ by uniform sampling over the boundary of $\wh C$. Lastly, the interval $\xi_j^*(\wh C)$ defined in Proposition~\ref{lm:converage} controls the worst case over $\wh C$. A detailed comparison among the four methods will be conducted in Section \ref{sec:randsim}. In general, the oracle interval $\xi_j(\mu_0,\alpha)$ reaches the nominal coverage rate with the shortest interval length, the coverage of $\xi_j(\hmu,\alpha)$ tends to be lower than the desired level,
while $\xi_j^*(\wh C)$ seems too conservative. The interval $\xi_j(\wh C)$ reaches a good compromise between coverage and interval length. 

\begin{figure}[t]
	\centering
		\includegraphics[width=0.9\textwidth]{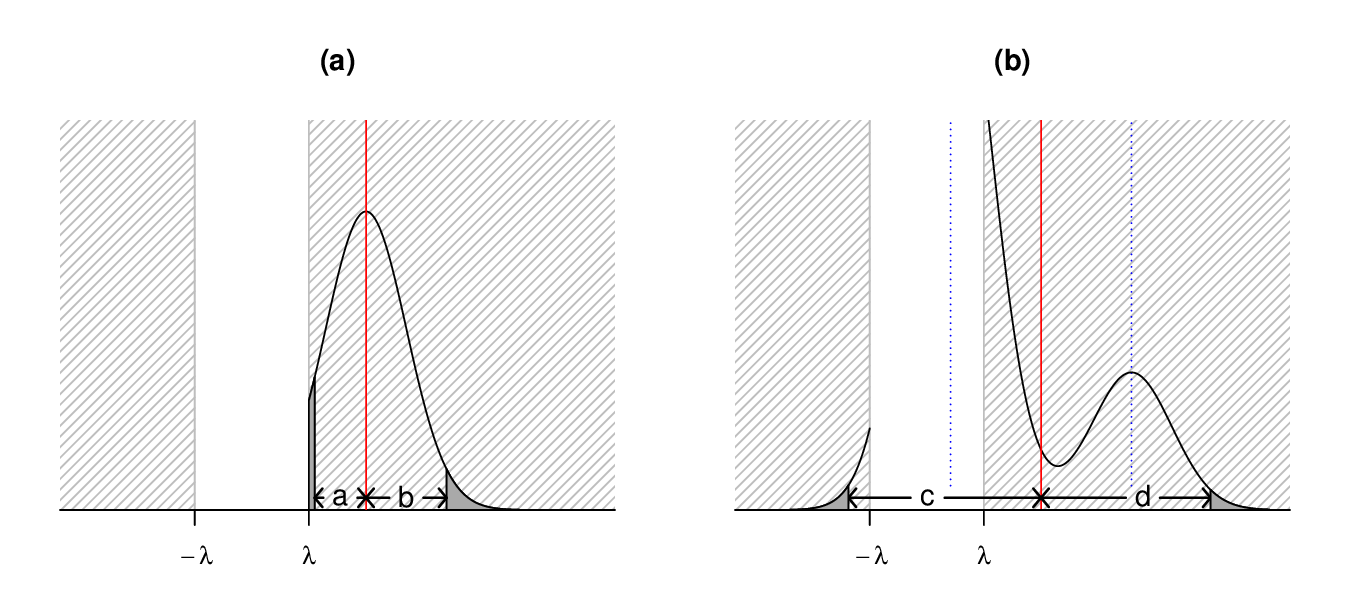}
\vspace*{-.2in}
	\caption{Illustration of the confidence intervals (a) $\xi(\hmu)$ and (b) $\xi(\wh C)$ in one-dimension. The red lines indicate $\hat \nu=\hbeta^{\text{LS}}$ and the two blue dotted lines indicate the boundaries of $\wh C$.}
	\label{fig:illust}
\end{figure}

Here, we illustrate the difference between $\xi_j(\hmu,\alpha)$ and $\xi_j(\wh C)$ for $p=1$, assuming $\|X_1\|^2= n$. In this case, the lasso $\hbeta=\sgn(\hbeta^{\text{LS}})(|\hbeta^{\text{LS}}|-\lambda)_+$, where $\hbeta^{\text{LS}}= X_1^\trans y/n$ is the least-squares estimate. The distribution of $\hbeta^{\text{LS}}$ given $\calA(y)=\{1\}$ is truncated to the intervals $(-\infty, -\lambda) \cup (\lambda,\infty):=T$ (shaded regions in Figure~\ref{fig:illust}). Write the two confidence intervals as $\xi(\hmu)$ and $\xi(\wh C)$, where $\hmu=X_1\hbeta^{\text{LS}}$ and $\wh C$ projected to $X_1$ is an interval $(\hbeta^{\text{LS}}-\Delta,\hbeta^{\text{LS}}+\Delta)$, centered at $\hbeta^{\text{LS}}$ between the two blue dotted lines in panel (b). Both $\xi(\hmu)$ and $\xi(\wh C)$ are centered at $\hat\nu=\hbeta^{\text{LS}}$, but with different end points. Let $\calT\dnorm_d(\mu,\Sigma,\calV)$ denote $\dnorm_d(\mu,\Sigma)$ truncated to the set $\calV$. The interval $\xi(\hmu)=[\hat \nu -b, \hat \nu +a]$ is constructed based on the quantiles of 
\begin{align*}
X_1^\trans y^*/n \mid \calA(y^*) = \{1\}\sim \calT\dnorm(\hbeta^{\text{LS}},\sigma^2/n,T),
\end{align*}
indicated by the dark tail regions in Figure~\ref{fig:illust}(a), where $y^* \sim \dnorm_n(\hmu, \sigma^2 \mathbf{I}_n)$. On the contrary, as shown in Figure~\ref{fig:illust}(b), the interval $\xi(\wh C) = [\hat \nu-d, \hat \nu+c ]$ is constructed from the quantiles of a mixture of two truncated normal distributions, i.e. $\calT\dnorm(\hbeta^{\text{LS}}\pm\Delta,\sigma^2/n,T)$, each centered at a boundary of the interval $\wh C$ (after being projected to $X_1$). If the true parameter $\beta_0\in(-\lambda,\lambda)$ is close to zero, then $\xi(\wh C)$ is likely to cover $\beta_0$ while the other interval $\xi(\hmu)$ will fail. In fact, the difficulty in post-selection inference largely  stems from such a situation in which some $\beta_{0j}$ is very close to zero and consequently the conditional distribution of $y$ given the selection event can change substantially with the mean $\mu_0$. Our method tackles this difficult problem by simulating from a mixture of such conditional distributions with mean $\tdmu$ randomized over a suitable neighborhood of $\mu_0$.

\section{Conditional sampling} \label{sec:condsamp}

In this section, we develop an MCMC sampler to draw $y^*$ such that $\calA(y^*) = A$, which is
the key conditional sampling step in our method. Our sampler is based on the idea of estimator augmentation.
So we first briefly review estimator augmentation for the lasso. 

\subsection{Estimator augmentation} \label{sec:EAlasso}
Let $\Psi = X^\trans X / n$. We start from the Karush-Kuhn-Tucker (KKT) condition for the lasso defined in \eqref{eq:lassodef},
\begin{eqnarray}\label{eq:lassoKKT}
\frac{1}{n}X^\trans y = {\Psi}\hbeta + \lambda W S, 
\end{eqnarray}
where $W = \diag(w_i)_{i=1}^p$ and $S$ is the subgradient of the $\ell_1$ norm at $\hbeta$:
\begin{align}\label{eq:subgdef}
\begin{cases}
S_i = \sgn(\hbeta_i)  \quad\text{if } \hbeta_i\ne 0, \\
S_i  \in [-1, 1] \quad\text{otherwise}.  \\
\end{cases}
\end{align}
\cite{zhou:14} inverted the KKT condition to find the sampling distribution of the so-called augmented estimator, $(\hbeta, S)$, linking its density to that of $X^\trans y$. Let $U = \frac{1}{n}X^\trans \varepsilon$ and $\Theta = (\hbeta_\calA, S_\calI)$, where both $\calA=\supp(\hbeta)$ and $\calI=\N_p \setminus \calA$ are random as functions of $\hbeta$. Note that $(\Theta,\calA)$ gives a parameterization of $(\hbeta,S)$ due to the definition of the subgradient $S$. The KKT condition can be rewritten,
\begin{eqnarray}\label{eq:UandH} 
U= {\Psi}\hbeta + \lambda WS - \frac{1}{n}X^\trans\mu_0 
:= H(\Theta, \calA ; \mu_0, \lambda).  
\end{eqnarray}
{Unless otherwise noted, we assume that $n < p$ and $X$ has  full row rank, i.e. $\rank(X) = n$.
Under this setting, the vector $U\in \row(X)$, an $n$-dimensional subspace of $\R^p$. Let $V$ be a $p \times p$ orthogonal matrix such that (i) the first $n$ columns of $V$, indexed by $R = \{1, \ldots, n\}$, consist of $n$ orthonormal eigenvectors associated with the positive eigenvalues of $\Psi$, and (ii) the last $p-n$ columns, indexed by $N=\{n+1, \ldots, p\}$, are a collection of orthonormal vectors that forms a basis of $\nul(X)$.}
Then $U$ can be re-expressed by its coordinates with respect to $V_R$ as 
\begin{align}\label{eq:distrnR}
R=V_R^\trans U\sim\dnorm_n(0,\sigma^2{\Lambda}/n),
\end{align}
where ${\Lambda} = \diag(\Lambda_i)_{i=1}^n$ and $\Lambda_i$'s are the positive eigenvalues of $\Psi$. Let $f_R$ be the density of $R$.
Equation~\eqref{eq:UandH} enforces a set of constraints on the $p$-vector $S$, i.e. $V_N^\trans WS=0$, since $WS$ must lie in $\row(X)$.
Denote the value of the random vector $(\hbeta,S)$ by $(b,s)$ and the corresponding value of $(\Theta,\calA)$ by $(\theta,A)=(b_A,s_I,A)$, where $A \subset \N_p$ and $I = \N_p \setminus A$. Then $\theta$ must satisfy the constraints
\begin{eqnarray} 
V_{AN}^{\trans}W_{AA}\sgn(b_A) + V_{IN}^{\trans}W_{II}s_{I}  &=& 0,  \label{eq:constraint1} \\
\|s_I\|_\infty &\leq& 1. \label{eq:constraint2}
\end{eqnarray}
Let $q = |A| \leq n$ (Remark \ref{rm:general}). Differentiating \eqref{eq:constraint1}, one sees that $ds_I \in \nul(V_{IN}^{\trans}W_{II})$, which is an $(n-q)$-dimensional subspace of $\R^{|I|}$. Thus $s_I$
can be parameterized by $s_F\in \R^{n-q}$ such that $ds_I = B(I)ds_F$, where $F$ is a size-$(n-q)$ subset of $I$ and $B(I) \in \R^{|I| \times (n - q)}$ is an orthonormal basis of $\nul(V_{IN}^{\trans}W_{II})$. 
Under mild conditions the $H$ defined in \eqref{eq:UandH} is a bijection (Lemma 3 in \cite{zhou:14}), which is used to derive the distribution for $(\Theta,\calA)$ from the density $f_R$. 
To ease our notation, let $d\theta\defi db_A ds_F$ be a differential form of order $n$. \cite{zhou:14} showed that the distribution of $(\Theta, \calA)$ can be represented by such $n$-forms:

\begin{theorem}[Theorem 2 in \cite{zhou:14}]\label{thm:lasso}
Assume $p>n$, the columns of $X$ are in general position and every $(p-n)$ rows of $V_N$ are linearly independent. If $y\sim\dnorm_n(\mu_0,\sigma^2\bfI_n)$ and $\lambda>0$, then the joint distribution of $(\Theta, \calA)=(\hbeta_\calA,S_\calI,\calA)$ is given by
\begin{eqnarray}\label{eq:joint}
\Prob_{\Theta,\calA}(d\theta, A) =  f_R\big(V_R^\trans H(\theta, A ; \mu_0, \lambda); \sigma^2 \big)|\det T(A; \lambda)|  d\theta
\end{eqnarray}
for $(\theta,A)$ satisfying \eqref{eq:constraint1} and \eqref{eq:constraint2}, where $f_R(\bullet;\sigma^2)$ is the density of the distribution in \eqref{eq:distrnR} and $T(A;\lambda) = [V_R^\trans {\Psi}_A | \lambda V_{IR}^\trans W_{II}B(I)] \in \R ^{n \times n}$.
\end{theorem}

\begin{remark}\label{rm:general}
The right side of \eqref{eq:joint} defines a joint density of $(\Theta,\calA)$ with respect to the parameterization $(b_A,s_F)$ for $\theta$. This density will be used to develop an MCMC algorithm for our conditional sampling step.
Mild assumptions are imposed on the design matrix. If the entries of $X$ are drawn from a continuous distribution over $\R^{n\times p}$, these assumptions will hold almost surely.
These assumptions also guarantee that the lasso solution is unique and $|\calA|\leq n$ for every $y \in \R^n$ \citep{tibs:13}. See \cite{Ewald20} for more recent discussion on the uniqueness of the lasso estimator.
Hereafter, when conditioning on $\calA=A$ we always assume that $|A|\leq n$.
\end{remark}

\subsection{The target distribution}

As an immediate consequence of Theorem \ref{thm:lasso}, we can obtain a density for the conditional distribution $[\hbeta_A, S_I \mid \calA=A]$ for a fixed subset $A$, which is directly related to our target conditional sampling problem.
\begin{corollary}\label{cor:conden}
Under the same assumptions of Theorem~\ref{thm:lasso}, the conditional distribution $[\hbeta_A, S_I \mid \calA=A]$
is given by 
\begin{eqnarray}\label{eq:cond}
\Prob_{\Theta|\calA}(d\theta |A) \propto  f_R\big(V_R^\trans H(\theta, A ; \mu_0, \lambda); \sigma^2 \big) d\theta:=\pi(\theta\mid A;\mu_0,\sigma^2,\lambda) d\theta,
\end{eqnarray}
where $\theta=(b_A,s_I)$ satisfies the constraints in \eqref{eq:constraint1} and \eqref{eq:constraint2}.
\end{corollary}
The conditional distribution $[\hbeta_A, S_I \mid \calA=A]$ has an especially simple density $\pi(\theta\mid A)$,
which is just an $n$-variate density with respect to a fixed parameterization $(b_A,s_F)\in\R^n$ as the active set $\calA$ is fixed to $A$ and the set $F$ only depends on $A$. See Corollary 1 in \cite{zhou:14} for a more detailed discussion on the truncated Gaussian nature of $\pi$. 

Given the density in Corollary~\ref{cor:conden}, we develop a Metropolis-Hastings (MH) sampler to draw $(\hbeta_A, S_I)$ given the fixed active set, $\calA = A$. This will achieve our goal of sampling $[y\mid \calA(y)=A]$ because of the following result:
 
\begin{theorem}\label{thm:condsample}
Suppose the assumptions of Theorem~\ref{thm:lasso} hold and $(\hbeta^*_A,S^*_I)$ follows the distribution in \eqref{eq:cond}. Then we have 
\begin{eqnarray}\label{eq:represent}
[y\mid \calA(y)=A]=\Big[X_A\hbeta^*_A + n\lambda(X^\trans)^+\big\{W_A\sgn(\hbeta^*_A) + W_I S^*_I\big\}\Big],
\end{eqnarray}
and consequently,
\begin{eqnarray}\label{eq:XAy}
[X_A^+ y\mid \calA(y)=A]  = \left[\hbeta^*_A + n\lambda(X_A^\trans X_A)^{-1}W_{AA}\sgn(\hbeta^*_A)\right].
\end{eqnarray}
\end{theorem}

As described in Algorithm \ref{alg:CI}, we wish to draw $[y^* \mid \calA(y^*)=A]$ for $y^*\sim \dnorm_n(\tdmu,\sigma^2\bfI_n)$. 
Once we have drawn $(\hbeta^*_A, S^*_I)$ from the density $\pi(\theta\mid A;\tdmu,\sigma^2,\lambda)$ \eqref{eq:cond}, we can easily obtain a sample of $y^*$ by \eqref{eq:represent}, which follows the target conditional distribution. Note that only $X_A^+ y^*$ is needed in \eqref{eq:qhatC} and \eqref{eq:DistJointqhatC}, which can be calculated directly with \eqref{eq:XAy}.

To provide an intuitive understanding of the conditional distributions in \eqref{eq:cond} and \eqref{eq:represent}, let us consider a simple example with $n>p=2$, $\Psi = \mathbf{I}_2$, $\mu_0=X\beta_0$ and $A = \{1\}$. In this low-dimensional setting, $\nul(X)=\{0\}$ and thus the constraint \eqref{eq:constraint1} is trivially satisfied for all $s\in\R^2$ (as $V_N=0$). As shown in Figure~\ref{fig:sampleSpace}(a), the sample space of $(\hbeta_A,S_I)=(\hbeta_1,S_2)$ is 
\begin{align*}
(-\infty,0) \times [-1,1] \cup (0,\infty)\times[-1,1] :=\Omg_{-1} \cup \Omg_{1},
\end{align*}
which is an essentially connected set (i.e. having a connected closure). Since $\Psi=\bfI_2$, we may choose $V_R=\bfI_2$, whose columns form an orthonormal basis for $\row(X)=\R^2$, and under this choice $f_R$ is the density of $\dnorm_2(0,\sigma^2\bfI_2/n)$. The contours of $[(\hbeta_1,S_2)\mid \calA=\{1\}]$ are shown in Figure~\ref{fig:sampleSpace}(a). It is easier to understand this distribution if further conditioning on $\sgn(\hbeta_1)=s_1$:
\begin{align*}
(\hbeta_1,S_2)\mid (\calA=\{1\},\sgn(\hbeta_1)=s_1) \sim \calT\dnorm_2(\mu(s_1),\Sigma,\Omg_{s_1}),
\quad\quad s_1\in\{-1,1\},
\end{align*}
which is a bivariate normal distribution truncated to $\Omega_{s_1}$ for each $s_1$. The mean and covariance matrix are
\begin{eqnarray*}
\mu(s_1) = \begin{pmatrix} \beta_{01} - \lambda s_1 \\\ \beta_{02}/\lambda \end{pmatrix}, \quad
\Sigma= \frac{\sigma^2}{n}\begin{pmatrix} 1 & 0\\ 0 & 1/\lambda^2 \end{pmatrix}.
\end{eqnarray*}
The two sets of contours in panel (a), separated by the line segment $\{0\}\times [-1,1]$, correspond to the two truncated normal distributions with different centers, $\mu(1)$ and $\mu(-1)$.
Figure~\ref{fig:sampleSpace}(b) plots the contours of the conditional distribution of $X^\trans y/n$ given $\calA=\{1\}$, which is a bivariate normal distribution $\dnorm_2(\beta_0,\sigma^2 \bfI_2/n)$ truncated to the union of two disconnected regions, 
\begin{align*}
(-\infty,-\lambda)\times [-\lambda,\lambda] \cup (\lambda,\infty)\times [-\lambda,\lambda].
\end{align*}
The contrast between the two sample spaces illustrates the potential advantage in designing Monte Carlo algorithms in the space of the augmented estimator $(\hbeta_A,S_I)$ over the space of $y$.

\begin{figure}[!t]
	\centering
		\includegraphics[width=0.9\textwidth]{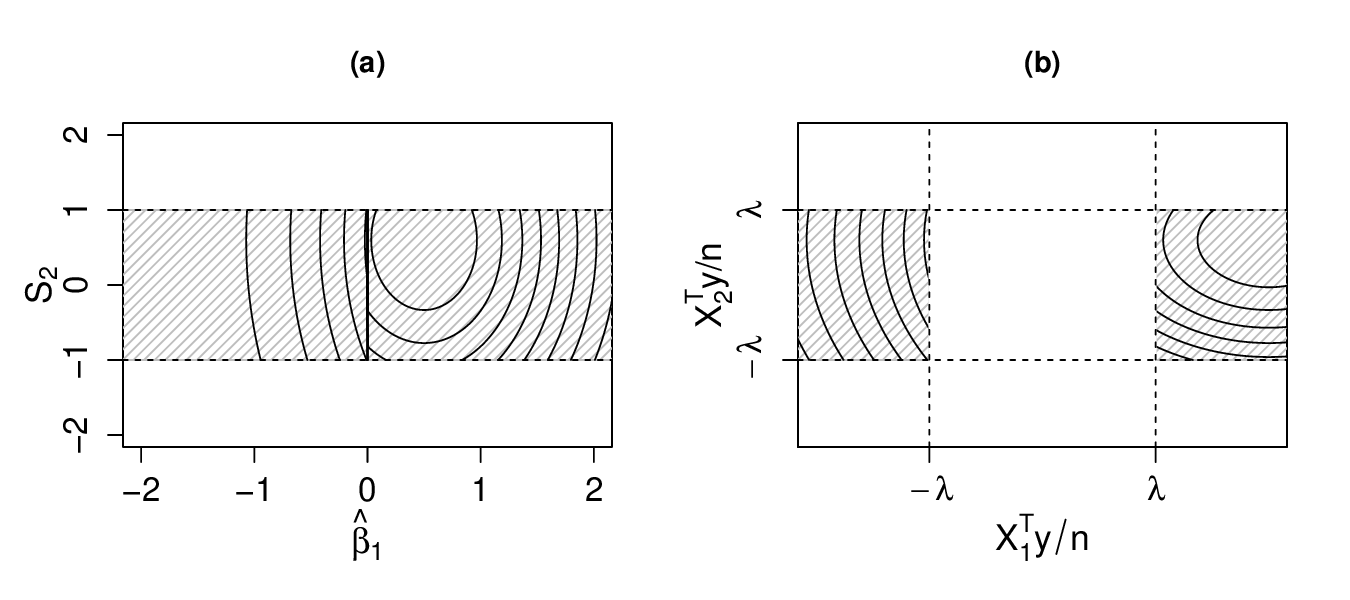}
\vspace*{-.2in}
	\caption{The conditional distributions of (a) $(\hbeta_1, S_2)$ and (b) $X^\trans y/n$ given $A= \{1\}$ for $p=2$.}
	\label{fig:sampleSpace}
\end{figure}

\subsection{A Metropolis-Hastings sampler}\label{sec:MHLS}

In what follows, we describe our MH sampler in detail. For notational brevity, write the target density as $\pi(\theta)\equiv\pi(\theta\mid A; \tdmu, \sigma^2, \lambda)$ hereafter, where the space of $\theta=(b_A,s_I)$ is defined by \eqref{eq:constraint1} and \eqref{eq:constraint2}. These constraints must be satisfied in each step of the sampling process, which presents a technical challenge for our Monte Carlo algorithm. We adopt a coordinate-wise update of $\theta$. Let $\theta^{(t)}$ be the value of $\theta$ at the $t$-th iteration. 
After proposing a new value $\theta^\dagger_i$ for its $i$-th component, the MH ratio is computed as
$$
\zeta = \min \left\{1, \frac{\pi(\theta^\dagger)}{\pi(\theta^{(t)})}\frac{q(\theta^\dagger, \theta^{(t)})}{q(\theta^{(t)}, \theta^\dagger)} \right\},
$$
where $q(\theta^{(t)}, \theta^\dagger)$ is the transition kernel of the proposal $\theta^\dagger$ given $\theta^{(t)}$. 
 If $\theta^\dagger$ is accepted, let $\theta^{(t+1)} = \theta^\dagger$. Otherwise, we reuse the previous state, i.e. $\theta^{(t+1)} = \theta^{(t)}$. 

We first derive explicit expressions for the feasible region of $\theta$ defined by \eqref{eq:constraint1} and \eqref{eq:constraint2}.
For the sake of notational simplicity, put
\begin{align*} 
G = V_{IN}^{\trans}W_{II} \in \R^{(p-n) \times |I|},\quad\quad u =u(s_A)= -V_{AN}^{\trans}W_{AA}\sgn(b_A) \in \R^{p-n}, 
\end{align*}
where $s_A=\sgn(b_A)$, and rewrite \eqref{eq:constraint1} as $G s_I = u$. Recall $|I|=p-|A|$ and $q=|A|$.
Since $p-n$ constraints are imposed on $s_I$, there are only $n-q$ free coordinates in $s_I$. Partition $s_I$ into free and dependent components and denote them by $s_F\in\R^{n-q}$ and $s_D\in\R^{p-n}$, respectively. Partition the columns of $G$ accordingly. Then \eqref{eq:constraint1} can be rewritten
\begin{equation} \label{eq:S_Iconstraint}
G_Fs_F + G_Ds_D=u \iff
s_D = G_D^{-1}(u-G_Fs_F),  
\end{equation}
which shows that $s_D$ is a function of $(b_A,s_F)\in\R^n$. Now the feasible region for $\theta$ can be equivalently defined by 
\begin{align}\label{eq:feasible}
\|s_F\|_{\infty}\leq 1,\quad\quad\|G_D^{-1}(u(s_A)-G_Fs_F)\|_\infty \leq 1.
\end{align}
Note that every time we update any component of $(b_A, s_F)$, $s_D$ needs to be updated accordingly via \eqref{eq:S_Iconstraint}. 
Below, we provide details about how to draw $(b_A, s_F)$, the free coordinates of $\theta$,
given the current value $(b_A^{(t)},s_F^{(t)})$. We assume that $(b_A^{(t)},s_F^{(t)})$ is feasible and satisfies \eqref{eq:feasible}.
 
For the active coefficients $b_A$, a normal distribution is used as the proposal, 
$$b_i^\dagger | b_i^{(t)}  \sim \calN(b_i^{(t)}, \tau_i^2),\quad\quad i \in A.$$
By using a symmetric proposal distribution, the MH ratio becomes the ratio of the target densities only,
\begin{equation} \label{eq:P1acceptancerate}
\zeta = \min \left\{1, \frac{\pi(\theta^\dagger)}{\pi(\theta^{(t)})}\right\}
= \min \left\{1, \frac{f_R\big(V_R^\trans H(\theta^\dagger, A ; \tdmu, \lambda); \sigma^2 \big) }
{f_R\big(V_R^\trans H(\theta^{(t)}, A; \tdmu, \lambda); \sigma^2 \big) } \right\}. 
\end{equation}
Under this proposal, $s_F^\dag=s_F^{(t)}$ is unchanged. 
If $\sgn(b_i^{\dagger}) = \sgn(b_i^{(t)})$, then $s_A^\dag=s_A^{(t)}$.
Consequently, $\theta^\dag$ satisfies the constraints in \eqref{eq:feasible} and thus is feasible.
If $\sgn(b_i^{\dagger}) \neq \sgn(b_i^{(t)})$, then 
$s_A^\dagger$ is different from $s_A^{(t)}$, with the $i$-th element replaced by $\sgn(b_i^{\dagger})$. We need to verify the second inequality in \eqref{eq:feasible}. Let $u^\dagger = u(s_A^\dagger)$. 
If $\|G_D^{-1}(u^\dagger-G_Fs^\dagger_F)\|_\infty \leq 1$, then $\theta^\dagger$ is feasible and we compute the MH ratio as in \eqref{eq:P1acceptancerate}. Otherwise, we move to the next component in $A$. 

When updating each component in $s_F$, denoted by $(s_F)_k$, it would be inefficient to use a naive proposal distribution, such as ${\cal U}(-1,1)$, since it does not guarantee every component of $s_D^{\dagger}$ will stay in $[-1,1]$.  A better approach is to compute the feasible range of $(s_F)_k$. Holding $s_A$ and 
$(s_F)_{-k}$ as constants,
the second inequality in \eqref{eq:feasible} defines $2(p-n)$ linear constraints on $(s_F)_{k}$,
from which the feasible range of $(s_F)_{k}$, $[LB_k, UB_k]$, can be computed,
\begin{eqnarray}
LB_k &=& \max\bigg\{-1, M^{-1} \big(- \mathbf{1}_{[p-n]} + G_D^{-1}u - (G_D^{-1}G_F)_{ -k}(s_F)_{-k}  \big)\bigg\},\label{eq:LB}\\
UB_k &=& \min\bigg\{1, M^{-1} \big( \mathbf{1}_{[p-n]} + G_D^{-1}u - (G_D^{-1}G_F)_{ -k}(s_F)_{-k}  \big)\bigg\},\label{eq:UB}
\end{eqnarray}
where $M = \diag\big((G_D^{-1}G_F)_{k}\big)$ is a $(p-n)\times(p-n)$ diagonal matrix having the $k$-th column of $G_D^{-1}G_F$ as its diagonal elements. Calculate $LB_k^{(t)}$ and $UB_k^{(t)}$ with $u^{(t)}=u(s_A^{(t)})$
and $(s_F^{(t)})_{-k}$. Note that $LB_k^{(t)} < UB_k^{(t)}$ since the current value $(s^{(t)}_F)_{k}\in[LB_k^{(t)}, UB_k^{(t)}]$ by assumption.
Propose $(s_F)_{k}^\dagger$ from $\mathcal{U}(LB_k^{(t)},UB_k^{(t)})$ and compute $s_D^\dagger$ by plugging  $s_F^\dagger$ and $u^\dagger=u^{(t)}$ in \eqref{eq:S_Iconstraint}.
Because $[LB_k, UB_k]$ does not depend on the current value of $(s_F)_{k}$, this proposal is symmetric, 
$q(\theta^\dagger, \theta^{(t)})=q(\theta^{(t)}, \theta^\dagger).$
Therefore, the MH ratio again reduces to \eqref{eq:P1acceptancerate}.

\setcounter{algorithm}{2}
\begin{algorithm}[!t]
\caption{$MH(\tdmu,\sigma,\lambda)$} 
\label{alg:MH} 
\begin{algorithmic}[1] 
	\STATE Given $(b_A, s_F)^{(t)}$,
	\FOR{$i\in A$}
		\STATE draw $b^{\dagger}_i \sim \calN(b^{(t)}_i , \tau^2_i)$.
		\IF{$\sgn(b_i^\dagger) \neq \sgn(b_i^{(t)})$ and $\theta^\dagger$ is infeasible}
		\STATE continue.
		\ELSE
		\STATE $b^{(t+1)}_i \leftarrow b^{\dagger}_i$ with probability $\zeta$; otherwise $b^{(t+1)}_i \leftarrow b^{(t)}_i$. 
		\STATE $b_{A \setminus i}^{(t+1)} \leftarrow b_{A \setminus i}^{(t)}$, $s_F^{(t+1)} \leftarrow s_F^{(t)}$, $t \leftarrow t+1$.
		\ENDIF
	\ENDFOR
	\FOR{$k \in \N_{|F|}$}
		\STATE compute $LB_k^{(t)}$ and $UB_k^{(t)}$ by \eqref{eq:LB} and \eqref{eq:UB}.
		\STATE draw $(s_F)^\dagger_k \sim {\cal U}(LB_k^{(t)},UB_k^{(t)})$.
		\STATE $(s_F)_k^{(t+1)} \leftarrow (s_F)_k^{\dagger}$ with probability $\zeta$; otherwise $(s_F)_k^{(t+1)} \leftarrow (s_F)_k^{(t)}$.
			\item $b_{A}^{(t+1)} \leftarrow b_{A}^{(t)}$, $(s_F)_{-k}^{(t+1)} \leftarrow (s_F)_{-k}^{(t)}$, $t\leftarrow t+1$.
	\ENDFOR
\end{algorithmic}
\end{algorithm}

Putting the above pieces together we present the MH sampler in Algorithm \ref{alg:MH}. To highlight its dependency on $(\tdmu, \sigma, \lambda)$, we denote this algorithm by $MH(\tdmu,\sigma,\lambda)$.

\subsection{Examples}

Using a small dataset of size $(n,p) = (5,10)$, we compared our MH sampler with parametric bootstrap which provided the ground truth here. We estimated $\mu_0$ by $\tdmu = X\hbeta$, where $\hbeta$ is the lasso estimate. 
The active set chosen by the lasso was $A=\{6,9\}$. In parametric bootstrap, we simulated $y^*\sim\dnorm_n(\tdmu,\sigma^2\bfI_n)$ and found the lasso solution $\hbeta(y^*)$.
If the support of $\hbeta(y^*)$ was indeed $\{6,9\}$, the sample $\hbeta(y^*)$ would be accepted.
We ran this bootstrap method until we accepted $10,000$ samples whose active set $\calA(y^*)=A$.
This is essentially a naive rejection sampling method. 
Note that the bootstrap is only applicable for this small dataset. Even for such a small dataset, the number of bootstrap samples simulated in order to obtain 10,000 samples was $1.5 \times 10^5$, i.e., the acceptance rate was only $6.67\%$. This demonstrates the necessity of our MH sampler for this conditional sampling problem.
The MH sampler was then used to draw $20,000$ samples. See Figure \ref{fig:MCMCvalidity} for a comparison between the samples generated by the two methods. It can be seen from the scatter plots that the results of our MH sampler were very close to the exact samples generated by the bootstrap, providing a numerical validation of our algorithm.

\begin{figure}[!t]
	\centering
		\includegraphics[width=0.9\textwidth]{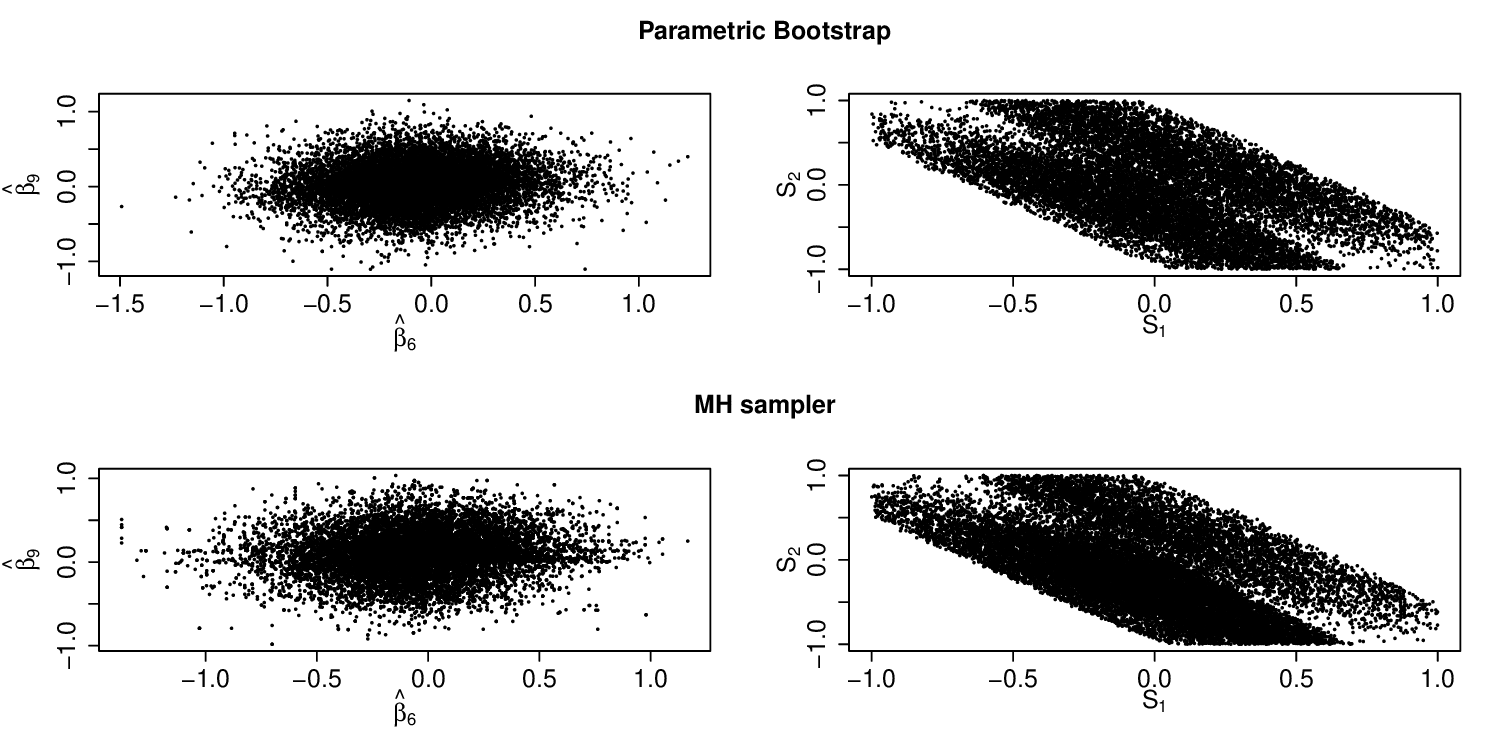}
\vspace*{-.2in}
	\caption{Comparison between bootstrap samples (top) and MH samples (bottom). The left column shows the scatter plot of $(\hbeta_6,\hbeta_9)$ while the right column shows the scatter plot of $(S_1, S_2)$. These are the first two components in $A$ and $I$, respectively. }
	\label{fig:MCMCvalidity}
\end{figure}

We present a quick visualization of the MH samples on a bigger dataset of size $(n,p)=(50,100)$. See Figure \ref{fig:ACF} for summary plots of the samples for the first two active coefficients, $\hbeta_1$ and $\hbeta_4$. The autocorrelation plots and the sample path plots show that our MH sampler was quite efficient with a fast decay in autocorrelation.

\begin{figure}[!t]
	\centering
		\includegraphics[width=0.9\textwidth]{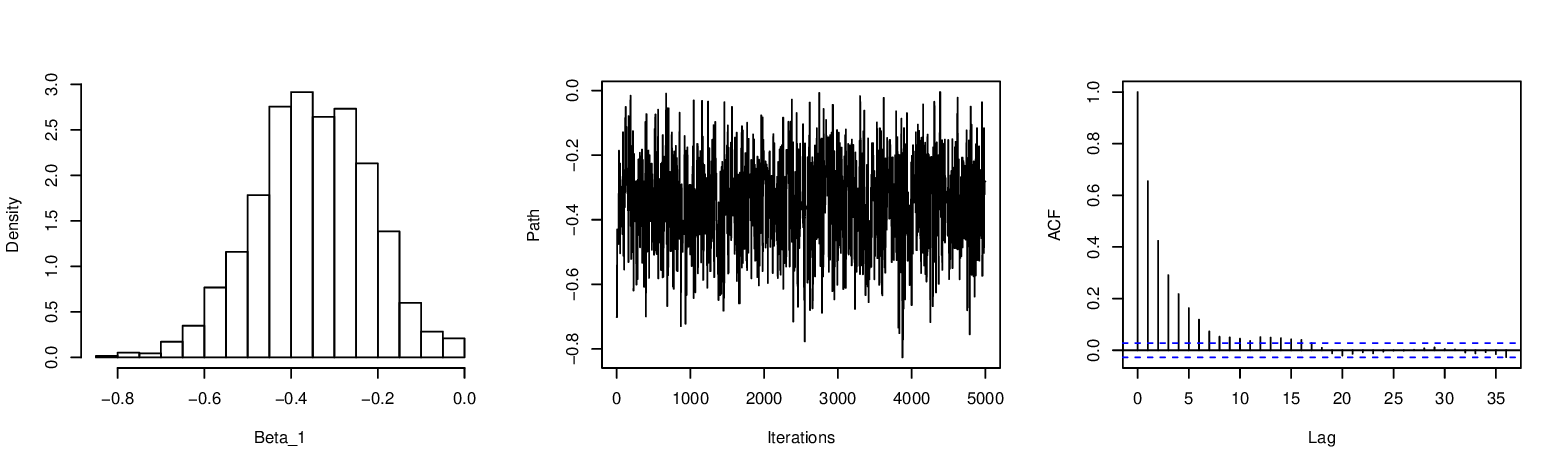}
				\includegraphics[width=0.9\textwidth]{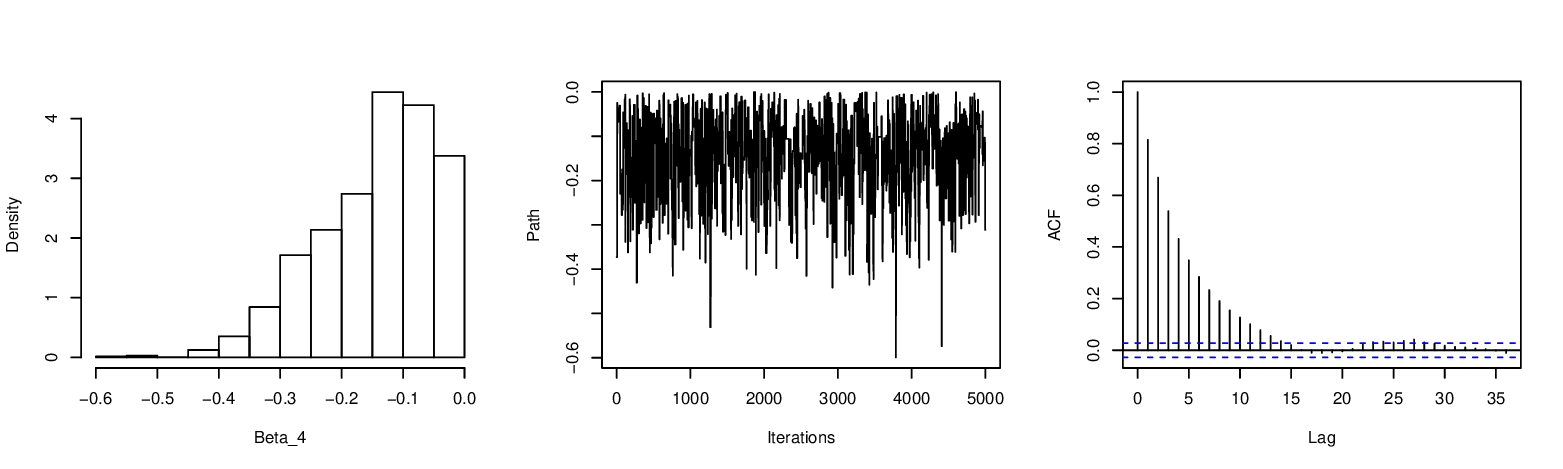}
\vspace*{-.2in}
	\caption{Summary plots for $\hbeta_1$ and $\hbeta_4$ from our MH sampler when $A = \{1,4,5\}$: histogram, sample trace, and autocorrelation function. }
	\label{fig:ACF}
\end{figure}

\section{Numerical results} \label{sec:numrslt}
In this section, we examine the performance of our method by providing simulation results under various settings.
In Section~\ref{sec:randsim}, we show the effectiveness of the proposed randomization of the plug-in estimate.   
Section~\ref{sec:senslbd} examines the robustness of our method with regard to the lasso tuning parameter $\lambda$. 
In Section~\ref{sec:leesim}, our confidence intervals are compared with those built by Lee's method. 
Section~\ref{sim:JS} provides simulation results for the construction of confidence sets by our method.
A detailed case study is presented in Section \ref{sec:caseStudy} to clarify the differences between our method and Lee's method.

\subsection{The effect of randomization}\label{sec:randsim}
To see the effect of our randomization step, we compare four different confidence intervals defined in Section \ref{sec:rand}:
\begin{itemize}
\compresslist
\item[(1)] $\xi_j(\mu_0) = [\hat \nu_j - q_{j, 1-\alpha/2}(\mu_0), \hat \nu_j - q_{j,\alpha/2}(\mu_0)]$ (oracle);
\item[(2)] $\xi_j (\hmu) = [\hat \nu_j -  q_{j,1-\alpha/2} (\hmu),\hat \nu_j -  q_{j,\alpha/2} (\hmu)],~ \hmu \in \{\hmu_A, \hmu_S\}$;
\item[(3)] $\xi_j (\wh C) = [\hat \nu_j - q_{j,1-\alpha/4} (\wh C),\hat \nu_j -  q_{j,\alpha/4} (\wh C)], ~\wh C \in \{\wh C_A, \wh C_S\}$;
\item[(4)] $\xi_j^*(\wh C) = [\hat \nu_j - q^*_{j,1-\alpha/4} (\wh C), \hat \nu_j -  q^*_{j, \alpha/4} (\wh C)], ~\wh C \in \{\wh C_A, \wh C_S\}$.
\end{itemize}
Recall that $\hat \nu = X_A^+ y$ and see \eqref{eq:CImu}, \eqref{eq:defupmax}, \eqref{eq:deflowmax} and \eqref{eq:qhatC} for the definitions of $q_{j,\gamma}(\mu)$, $q_{j,\gamma}(\wh C)$ and $q_{j,\gamma}^*(\wh C)$. 
As we described in Section \ref{sec:rand}, there are two ways of constructing $\wh C$, with center $\hmu$. The subscripts $A$ and $S$ are used to distinguish the two methods. Again, for intervals (3) and (4),  $(\alpha/4, 1-\alpha/4)$-quantiles are used due to Bonferroni correction. 

Algorithm \ref{alg:CI} is used to construct interval (3). Likewise, for interval (4), we draw $\big\{\tdmu^{(i)}\big\}_{i=1}^K$ from $\partial \wh C$ and estimate $q_{j,\alpha/4}^*(\wh C)$ and $q_{j,1-\alpha/4}^*(\wh C)$ by
\begin{align*}
{q}_{j,\alpha/4}^*(\wh C)=\min_{1 \leq i \leq K}q_{j,\alpha/4}(\tdmu^{(i)}), \quad\quad
{q}_{j,1-\alpha/4}^*(\wh C)=\max_{1 \leq i \leq K}q_{j,1-\alpha/4}(\tdmu^{(i)}).
\end{align*}
We set $K = 20$, and given each $\tdmu^{(i)}$, we sampled 500 $y^*$'s, i.e. the total number of samples used for intervals (3) and (4) was $20 \times 500=10,000$. 
For a fair comparison, we fixed the number of samples to be $10,000$ for (1) and (2). 
Note that our MH sampler was used in all the four methods to draw from $[y^* \mid \calA(y^*) = A]$ and estimate the quantiles $q_{j,\gamma}(\mu)$.
Twenty datasets with $(n,p,A_0) = (50,100,\mathbb{N}_5)$ were simulated. The true coefficients $\beta_{0A_0}$ were drawn from $\mathcal{U}(-1,1)$. Each row of $X$ was independently sampled from $\calN_p(0, \Sigma)$. We considered three types of covariance matrix $\Sigma$ in this comparison:
\begin{itemize}
\item Identity (I): 
$
{ \Sigma} = {\bf I}_p,
$
\item Toeplitz (T): 
$
{ \Sigma}_{ij} = 0.5^{|i-j|},
$
\item Exponential Decay (ED):
$
{ \Sigma}_{ij}^{-1} = 0.4^{|i-j|}.
$
\end{itemize}
The significance level $\alpha$ was set to $0.05$ and $\sigma^2=1$ was assumed to be known. For each dataset, $\lambda$ was chosen by cross-validation with the one standard error rule. 

The following metrics are used to compare the results. For a subset $E \subset \N_p$ and confidence intervals, $\wh I_j$ for $j \in A$, we define power and coverage by averaging over the variables in the set $E$:
\begin{itemize}
\item[] Power $= \sum_{j\in E} \mathbb{P} \Big(0 \notin \wh I_j\Big) / |E|$,
\item[] Coverage $= \sum_{j\in E} \mathbb{P} \Big( \nu_j \in \wh I_j\Big) / |E|$.
\end{itemize}
A few informative choices for $E$ are $A$, $A_0 \cap A$ and $A_0^c \cap A$.
The set $A$ includes all the variables selected by lasso, while the sets $A_0 \cap A$ and $A_0^c \cap A$ separate the true positive and the false positive variables. We report in Table \ref{tb:qexp} 
the average coverage rate over variables
in each of the three sets and the power of detecting true positive variables $A_0\cap A$ for each of the four methods. 
We omit the subscript $j$ to simplify our notation and to indicate averaging over a subset of indices, such as $j\in A$.

The coverage rate of $\xi(\mu_0)$ was at the desired level while its average length was the shortest among all the methods. This is an obvious result, since the true parameter $\mu_0$ is assumed to be known (the oracle). 
Using a single plug-in estimate, $\xi(\hmu)$ produced shorter confidence intervals (CIs) compared to $\xi(\wh C)$ and $\xi^*(\wh C)$. However,
the coverage rate of $\xi(\hmu)$ was much lower than the nominal level, especially for $j \in A_0^c \cap A$. On the contrary, with randomized $\tdmu$ drawn from the confidence set $\wh C$, the CIs of $\xi(\wh C)$ achieved the desired coverage rate, which demonstrates the importance of our proposed randomization step. The intervals $\xi^*(\wh C)$ showed a similar effect as $\xi(\wh C)$, but they turned out to be the most conservative with overall coverage rates close to 1 and the longest interval lengths. In particular, for the set $A_0^c \cap A$ the average length of $\xi^*(\wh C_A)$ was much longer than the length of $\xi(\wh C_A)$. 

\begin{table}[!t]

\caption{Power and coverage rate for (1) $\xi(\mu_0)$ (oracle), (2) $\xi(\hmu)$, (3) $\xi(\wh C)$ and (4) $\xi^*(\wh C)$.} 
\begin{center}
\vspace*{.05in}
\begin{tabular}{llllll}
 \hline 
    \multirow{2}{*}{$\Sigma$} & \multirow{2}{*}{Method} &Power &  & Coverage &  \\ 
\cline{4-6}
&& $A_0 \cap A$ & $A$ & ${A_0} \cap A$ & $ A_0^c \cap A$ \\
  \hline
I			 		& (1)			& 1.000 & 0.960(0.379) & 0.960(0.471) & 0.962(0.202) \vspace{.02in} \\ 
					& (2$_A$)	& 0.980 & 0.882(0.493) & 0.880(0.555) & 0.885(0.373) \\ 
					& (3$_A$)	& 0.760 & 0.934(0.751) & 0.900(0.857) & 1.000(0.546) \\ 
					& (4$_A$)	& 0.800 & 1.000(0.881) & 1.000(0.874) & 1.000(0.895) \vspace{.02in} \\ 
					& (2$_S$)	& 0.880 & 0.618(0.479) & 0.600(0.520) & 0.654(0.401) \\ 
					& (3$_S$)	& 0.580 & 0.934(0.850) & 0.900(0.973) & 1.000(0.613) \\ 
					& (4$_S$)	& 0.740 & 1.000(0.975) & 1.000(0.980) & 1.000(0.965) \vspace{.05in} \\ 
T			 		& (1)			& 0.978 & 0.955(0.426) & 0.956(0.542) & 0.954(0.191) \vspace{.02in} \\ 
					& (2$_A$)	& 0.978 & 0.821(0.562) & 0.933(0.634) & 0.591(0.416) \\ 
					& (3$_A$)	& 0.711 & 0.970(0.852) & 0.956(0.963) & 1.000(0.624) \\ 
					& (4$_A$)	& 0.800 & 0.985(0.945) & 1.000(0.966) & 0.954(0.902) \vspace{.02in} \\
					& (2$_S$)	& 0.800 & 0.642(0.537) & 0.667(0.593) & 0.591(0.421) \\ 
					& (3$_S$)	& 0.489 & 0.970(0.962) & 0.956(1.106) & 1.000(0.668) \\ 
					& (4$_S$)	& 0.689 & 0.985(1.044) & 1.000(1.077) & 0.954(0.979) \vspace{.05in} \\ 
ED			 	& (1)			& 0.967 & 0.957(0.361) & 0.951(0.456) & 0.969(0.180) \vspace{.02in} \\ 
					& (2$_A$)	& 0.934 & 0.914(0.451) & 0.918(0.521) & 0.906(0.318) \\ 
					& (3$_A$)	& 0.754 & 0.957(0.678) & 0.934(0.802) & 1.000(0.441) \\ 
					& (4$_A$)	& 0.721 & 1.000(0.816) & 1.000(0.831) & 1.000(0.785) \vspace{.02in} \\ 
					& (2$_S$)	& 0.869 & 0.774(0.427) & 0.721(0.472) & 0.875(0.340) \\ 
					& (3$_S$)	& 0.639 & 0.957(0.752) & 0.934(0.879) & 1.000(0.511) \\ 
					& (4$_S$)	& 0.688 & 1.000(0.901) & 1.000(0.908) & 1.000(0.889) \vspace{.05in} \\ 
\hline
\end{tabular}
\label{tb:qexp}
\vspace*{-.1in}
\end{center}

\small{The subscripts $A$ and $S$ indicate two ways of estimating $\wh C$ and its center $\hmu$. The average length of confidence intervals is reported in the parentheses.}
\end{table}

Between the two ways of constructing $\wh C$, $\xi(\wh C_A)$ and $\xi(\wh C_S)$ had the same coverage rates. However, we observe that the average length of $\xi(\wh C_S)$ was longer than that of $\xi(\wh C_A)$, which reduced its power. Therefore, in the following numerical results, we choose to use $\xi(\wh C_A)$ only. See related discussion in Section \ref{sec:JointInference}.

\begin{figure}[!t]
	\centering
	\rotatebox{-90}{
		\includegraphics[width=.45\textwidth]{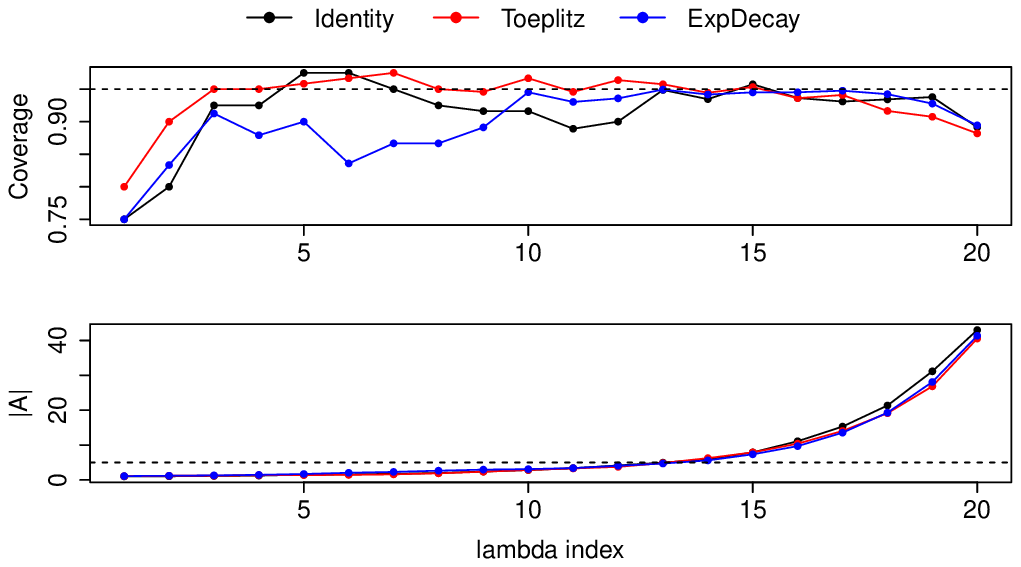}
	} 
	\caption{Sensitivity test for the choice of $\lambda$ with the datasets of $(n,p,A_0)=(50,100,\mathbb{N}_5)$. Each point is the average from 20 datasets.} \label{fig:lbdsens}
\end{figure}

\subsection{Sensitivity to $\lambda$}\label{sec:senslbd}

Using the same datasets from Section \ref{sec:randsim}, we ran more tests to examine how sensitive the coverage of ${\xi}(\wh C_A)$ is to $\lambda$, the tuning parameter of the lasso. We chose 
20 $\lambda$ values, equally spaced between $\|X^\trans y\|_\infty / n$ and $0$. Figure~\ref{fig:lbdsens} plots the coverage rate and the size of the active set $q = |A|$ against the index $i_\lambda$ of $\lambda$. Note that the $\lambda$ sequence is in decreasing order so that $q$ increases with $i_\lambda$. Each point in the plot is the average of 20 datasets. The coverage rate of ${\xi}(\wh C_A)$ was preserved around the desired level, indicated by the dashed line in the top panel, when the lasso active set was not extremely small or large. In fact, the coverage rate was well maintained around $95\%$ for $2 \leq q \leq 30$ $(7 \leq i_\lambda \leq19)$, while the size of the true active set $|A_0|= 5$ (the dashed line in the bottom panel). This shows that our method works well for a wide range of models selected by lasso.
The coverage rate was a little more sensitive to the choice of $\lambda$ for the exponential decay design than the other two designs. However, even for that case, when the size of active set $q\geq 3$ $(i_\lambda \geq 10)$, the coverage rate stayed around the desired level.

 One might worry about the low coverage rates for the first few and the last $\lambda$ values. However, these $\lambda$'s are either too small or too big to be chosen in practice. Recall that we choose $\lambda$ by cross-validation with the one standard error rule, denoted by $\lambda_{1se}$. The 5\% and 95\% percentiles of $i_{\lambda_{1se}}$ were 10.21 and 17.35, between which the performance of our method is seen to be very stable (Figure \ref{fig:lbdsens}). This analysis confirms the notion that our inference tool may be used in conjunction with a data-dependent turning of lasso to quantify the significance of a quite large set of selected features, as discussed in Remark~\ref{rmk:lbd}.

\subsection{Comparison on individual intervals}\label{sec:leesim}

\begin{table}[!t]
\begin{center}
\caption{Comparision between (a) $\xi (\wh C_A)$ and (b) Lee's method.}\label{tb:leecomparison1}
\vspace*{.1in}
\small{
\begin{tabular}{llllllllll}
 \hline 
\multirow{2}{*}{$(n,p)$} & \multirow{2}{*}{$A_0$} & \multirow{2}{*}{$\Sigma$} & \multirow{2}{*}{Method} &Power &  & Coverage &  & \multirow{2}{*}{$\mathbb{P}_\infty$} \\ 
\cline{6-8}
&&&& $A_0 \cap A$ & $A$ & ${A_0} \cap A$ & $ A_0^c \cap A$ \\
 \hline
$(100, 200)$		& $A_0^{(1)}$	& I 			& (a)		&	0.870 & 0.968(0.555) & 0.956(0.673) & 0.975(0.486) \\  
						&						& 				& (b) 	&	0.783 & 0.968(1.819) & 1.000(0.998) & 0.949(2.310)  &	1.6\%\\ 
						&						& T			& (a)		& 	0.836 & 0.953(0.590) & 0.934(0.721) & 0.978(0.416) \\ 
						&			 			& 		 		& (b)		&	0.820 & 0.944(1.562) & 0.951(0.882) & 0.935(2.617) &	9.3\%\\ 
						&						& ED			& (a)		&	0.870 & 0.931(0.501) & 0.896(0.606) & 0.981(0.349) \\ 
						&						& 				& (b)		&	0.818 & 0.946(1.251) & 0.922(0.824) & 0.981(1.887)  &	2.3\%\\
						&						& EC		   	& (a)		&	0.643 & 0.846(0.612) & 0.768(0.760) & 0.892(0.523) \\ 
						&						& 			   	& (b)		&	0.500 & 0.973(5.178) & 0.964(3.513) & 0.978(6.288) & 9.4\% \vspace*{.05in}\\ 		
						& $A_0^{(2)}$ 	&T			& (a)		&	0.911 & 0.914(0.539) & 0.889(0.657) & 0.928(0.476)  \\ 
						&			 			& 		 		& (b) 	&	0.822 & 0.969(1.618) & 0.956(0.994) & 0.976(1.984)  &	4.7\%\\ 
						&						& ED			& (a)		&	0.889 & 0.948(0.460) & 0.933(0.574) & 0.956(0.403)  \\  
						&						& 				& (b) 	&	0.822 & 0.993(2.733) & 0.978(1.301) & 1.000(3.457)  &	3.7\% \vspace*{.05in}\\ 	
$(200, 400)$		& $A_0^{(1)}$	& I 			& (a)		&	0.923 & 0.954(0.433) & 0.949(0.489) & 0.967(0.286) \\ 
						&						& 				& (b)		&	0.872 & 0.972(0.728) & 0.962(0.376) & 1.000(1.645)  &	0.0\%\\ 
						&						& T			& (a)		&	0.897 & 0.933(0.451) & 0.926(0.539) & 0.946(0.288) \\ 
						&			 			& 		 		& (b)		&	0.868 & 0.933(1.154) & 0.912(0.787) & 0.973(1.887)  &	2.9\%\\ 
						&						& ED			& (a)		&	0.909 & 0.928(0.362) & 0.896(0.468) & 0.968(0.230) \\ 
						&						& 				& (b)		&	0.779 & 0.978(1.443) & 0.987(0.956) & 0.968(2.127)  &	6.5\%\\ 		
						&						& EC		   	& (a)		&	0.841 & 0.926(0.492) & 0.921(0.671) & 0.929(0.379) \\ 
						&						& 			   	& (b)		&	0.714 & 0.975(3.481) & 0.984(2.078) & 0.970(4.422) & 6.2\% \vspace*{.05in}\\ 		
						& $A_0^{(2)}$ 	&T			& (a)		&	1.000 & 0.991(0.441) & 1.000(0.512) & 0.987(0.407) \\  
						&			 			& 		 		& (b) 	&	0.917 & 0.928(0.929) & 0.972(0.731) & 0.907(1.043)  &	13.5\%\\ 
						&						& ED			& (a)		&	0.971 & 0.927(0.383) & 0.914(0.441) & 0.932(0.355) \\  
						&						& 				& (b) 	&	0.857 & 0.972(0.945) & 0.943(0.510) & 0.986(1.162)  &	3.7\%\\ 		
\hline
\end{tabular}
}
\end{center}
\small{The numbers in the parentheses are the average length of the confidence intervals. For Lee's method, only finite intervals are used to compute the average length and $\mathbb{P}_\infty$ is the proportion of excluded infinite intervals.} 
\end{table}

In this section, $\xi(\wh C_A)$ is compared with Lee's method \citep{lee:etal:16} implemented in the \proglang{R} package {\bf selectiveInference}.
Datasets were simulated in the same way as in Section \ref{sec:randsim} but with two larger sizes,
$(n,p) \in \{(100,200), (200,400) \}$, and one more type of design matrix
\begin{itemize}
\item Equicorrelation (EC): 
${ \Sigma}_{ii} = 1$ and ${\Sigma}_{ij} = .7 ~(i \neq j)$.
\end{itemize}
Note that the correlation among predictors was the highest under this design. 
We also considered two different ways of placing true active coefficients. In the first case, the true active coefficients were placed together, i.e. $A_0^{(1)} = \{1, \cdots, 5\}$. In the second case, they were evenly spaced out, i.e. $A_0^{(2)} = \{1, p/5 +1, \cdots, 4p/5 + 1\}$.
The true active covariates were highly correlated with each other in the first case, while they were more correlated with other inactive covariates in the second case. See Table~\ref{tb:leecomparison1} for the comparison results. Note that the designs Identity and Equicorrelation were considered only with $A_0^{(1)}$, since the two ways of assigning $A_0$ are equivalent for these two designs.

The coverage rate of Lee's method was well-preserved at the nominal level, $1-\alpha$, in most cases. 
However, their method sometimes generated very wide or even infinite intervals with $\infty$ or $-\infty$ as the upper or lower bound. This happens when the conditional distribution $[(X_A^+ y)_j \mid \calA(y), (X_A^+ y)_{-j}]$ is truncated to a union of bounded intervals and the observed value of $(X_A^+ y)_j$ is close to one of its boundaries \citep{kiva:leeb:18}. See Section~\ref{sec:caseStudy} for a case study that exemplifies this issue.
The last column in Table \ref{tb:leecomparison1} reports the proportion of infinite intervals estimated by Lee's method. For example,  when $(n,p) = (200,400)$ and $A_0=A_0^{(2)}$ with the Toeplitz design, the proportion of infinite confidence intervals was $13.5\%$. The chance of encountering such an issue was already quite high but it would be even higher if we increased the confidence level. 

On the other hand, for most settings, our confidence intervals ${\xi}(\wh C_A)$ succeed to stay at the desired level while having much shorter average length than the intervals by Lee's method. For every setting except the case of $(n,p,A_0,\Sigma)=(100,200,A_0^{(1)},\text{EC})$, our coverage rates averaging over all $j \in A$ were higher than $0.9$ and very close to $0.95$. 
The average length of our intervals was uniformly shorter than that of Lee's method (after excluding infinite ones). The difference in the interval length was especially significant for the coefficients in $A \cap A_0^c$ and for the equicorrelation designs. 
For example, in Table~\ref{tb:leecomparison1} when $(n,p) = (200,400)$ and $A_0=A_0^{(1)}$ with the equicorrelation design, the average length of ${\xi}(\wh C_A)$ was 0.492, while the average length from Lee's method was $3.481$. This is extremely long considering the fact that we drew $\beta_{0j}$ from $\calU(-1,1)$ for $j \in A_0$. These long intervals failed to detect significant coefficients and thus resulted in a low power. 

\begin{figure}[!p]
	\centering
		\includegraphics[width=0.7\textwidth]{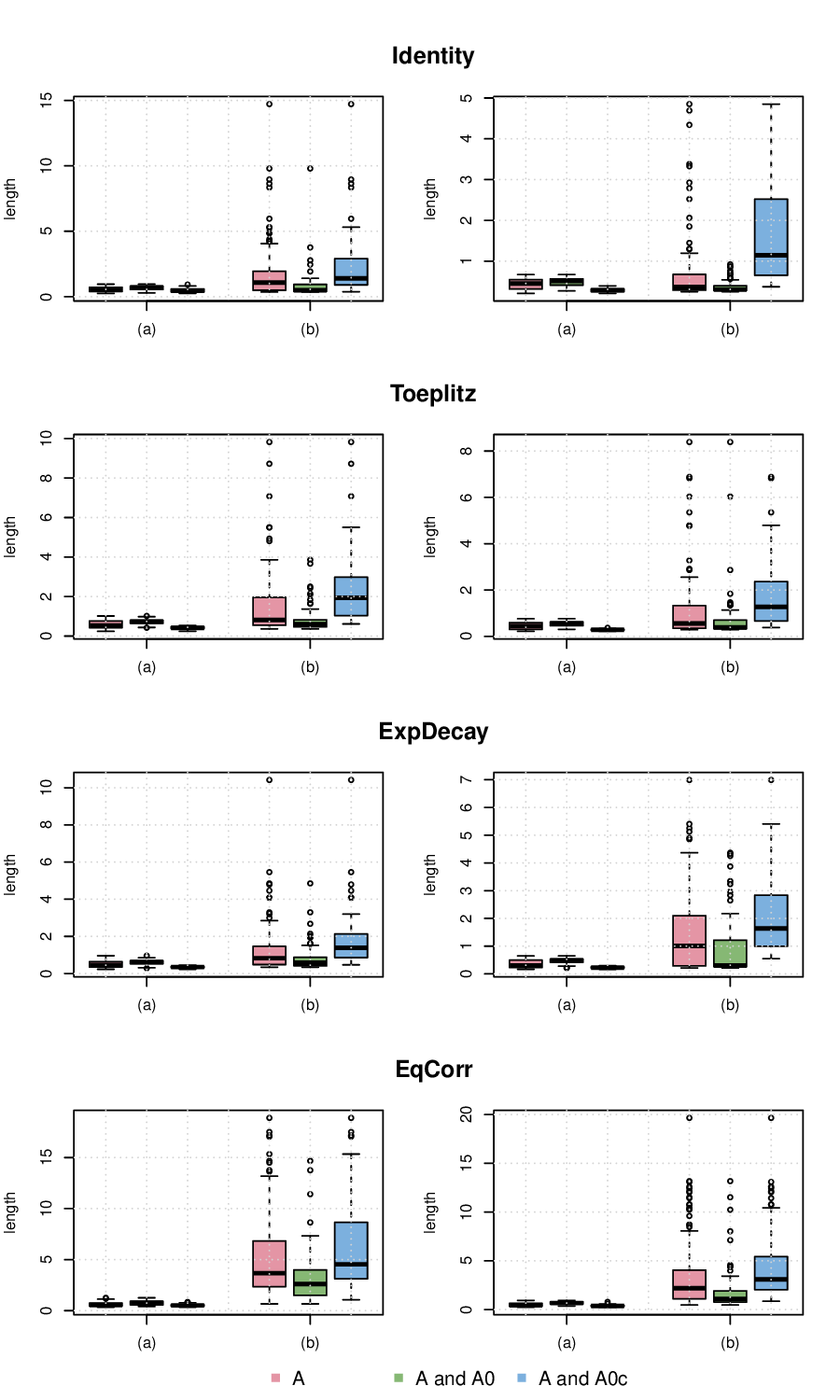}
	\caption{Comparison between (a) ${\xi}(\wh C_A)$ and (b) Lee's method when $A_0 = \{1, \cdots ,5\}$. The left and right columns are for $(n,p)=(100,200)$ and $(200,400)$, respectively. \label{fig:lengthBoxplotLee1}}
\end{figure}

In Figure \ref{fig:lengthBoxplotLee1}, we show box-plots of the interval lengths from the 20 datasets in each design with $A_0 = A_0^{(1)}$ for a closer view. The plots for the cases of $A_0 = A_0^{(2)}$ are almost identical and thus omitted here for brevity. Each box-plot reports the interval lengths for all variables in $A$, $A\cap A_0$ or $A\cap A_0^c$.
Consistent with the results from Table \ref{tb:leecomparison1}, the interval lengths of Lee's method were much larger than those of our method. In particular, under the equicorrelation design, the maximum length of our intervals was even smaller than the first-quartile length of Lee's method for all three sets of variables.
The length of our intervals is also much less variable than that of Lee's method for every case in the figure, which shows that our method is more consistent across different datasets. 
One can easily see that Lee's method produced a number of lengthy intervals, represented as isolated dots or outliers in a box-plot. These intervals are not informative at all.
Lastly, the difference between the two methods is most drastic for the set $A \cap A_0^c$, where the intervals from Lee's method can be 10 times longer than ours.
Note that by removing all the infinite intervals output by Lee's method from these plots, this comparison favors Lee's method.

\subsection{Comparison on joint confidence sets}\label{sim:JS}

We conducted further experiments to examine the performance of our method in constructing confidence sets for $\nu$. We generated results under three types of design, $\Sigma \in \{\text{T}, \text{ED}, \text{EC}\}$, and two data sizes, $(n,p) \in \{(100,200), (200,400)\}$, with $A_0$ fixed to $\{1, \ldots, 5\}$. 
Under each of these six settings, the same 20 datasets as in Section \ref{sec:leesim} were used. 
First, we constructed confidence sets $\xi_{[e_i, e_j]^\trans}(\wh C; \ell_\delta)$ \eqref{eq:JointqhatC} for each pair $(\nu_i, \nu_j)$, $i \neq j$, with $\ell_2$ norm and $\ell_\infty$ norm, i.e. $\delta\in\{2,\infty\}$. Then, we moved to confidence sets $\xi_{\mathbf{I}_q}(\wh C; \ell_\delta)$, $q=|A|$, for joint inference on $\nu$, again using the two norms, $\ell_2$ and $\ell_\infty$. Consequently, the confidence sets were either a sphere or a hypercube in $\R^q$. As we are not aware of any method specifically designed for joint inference after lasso selection, we compared our results with Lee's method using multiple test adjustment. To build a $1-\alpha$ confidence set for $\nu_B$, $|B|=d$, individual intervals $\wh I_k$, $k\in B$, were constructed  by Lee's method with an adjusted confidence level $1-\alp/d$, and then a confidence set was constructed as the Cartesian product of $\wh I_k, k\in B$.

The following metrics are used to evaluate constructed confidence sets. Recall $A$ is the set of selected variables and $q=|A|$.  For a positive integer $m$, let ${\cal B}(m) = \{B : B \subset \N_q, ~|B| = m\}$ index all size-$m$ subsets of $A$. We define coverage and power by averaging over sets of variables in ${\cal B}(m)$:
\begin{itemize}
\item[] Coverage $= \sum_{B\in {\cal B}(m)} \mathbb{P} \Big( \nu_B \in \xi_{e_B^\trans}(\wh C)\Big) / |{\cal B}(m)|$,
\item[] Power $= \sum_{B\in {\cal B}(m)} \mathbb{P} \Big(0 \notin \xi_{e_B^\trans}(\wh C)\Big) / |{\cal B}(m)|$.
\end{itemize}
For example, when considering $\xi_{\mathbf{I}_q}(\wh C)$, $m=q$ and $|{\cal B}(q)| = 1$.  For pairwise confidence sets, $m=2$ and $|{\cal B}(2)|=q(q-1)/2$. The volume and the diameter of a confidence set were recorded for comparison as well, where the diameter is defined as the maximum distance between any two points in the set. In short, a confidence set has a better performance if it has a higher power and a smaller volume or diameter, while achieving the nominal coverage rate.

\begin{table}[!t]
\caption{Coverage, power, and size of pairwise confidence sets}\label{tb:PairwiseqhatC}
\begin{center}
\begin{tabular}{lllllllll}
 \hline 
\multirow{2}{*}{} & {\multirow{2}{*}{Method}}  & \multicolumn{3}{c}{$(n,p) =(100,200)$}   & \multicolumn{3}{c}{$(n,p) =(200,400)$}  &  {\multirow{2}{*}{Overall}}  \\	 
& & T 	& ED & EC & T 	& ED & EC \\	 \hline 
{Coverage} 
& $\xi(\wh C; \ell_2)$ 		& 0.941	 &	0.971 &	0.965 &	0.912  & 0.912	& 0.937	 & 0.940 \\ 
& $\xi(\wh C; \ell_\infty)$	& 0.972 & 0.984 & 0.970 & 0.950 & 0.962 & 0.944 & 0.964 \\  \vspace*{.05in}
& Lee							& 0.924 & 0.973 & 0.997 & 0.972 & 0.989 & 0.992 & 0.975 \\  
{Power} 
& $\xi(\wh C; \ell_2)$ 		& 1.000 	& 0.983 	& 0.831 	& 1.000 	& 1.000 	& 0.986 & 0.967  \\ 
& $\xi(\wh C; \ell_\infty)$	& 0.986 & 0.975 & 0.769 & 1.000 & 0.992 & 0.986 & 0.951\\ \vspace*{.05in}
& Lee							& 0.718 & 0.558 & 0.138 & 0.736 & 0.558 & 0.568 & 0.546 \\ 
{Diameter} 
& $\xi(\wh C; \ell_2)$ 		& 0.885		&	0.734	&	1.045	&	0.644	&	0.488	&	0.792	& 0.765\\
& $\xi(\wh C; \ell_\infty)$	& 1.173 & 0.972 & 1.384 & 0.850 & 0.653 & 1.065 & 1.016 \\ \vspace*{.05in}
& Lee							& 2.854 & 2.301 & 10.221 & 1.935 & 3.493 & 6.251 & 4.509\\ 
{Volume} 
& $\xi(\wh C; \ell_2)$ 		& 0.627 & 0.431 & 0.883 & 0.331 & 0.193 & 0.521 & 0.498\\ 
& $\xi(\wh C; \ell_\infty)$	& 0.699 & 0.481 & 0.986 & 0.366 & 0.220 & 0.601 & 0.559 \\ \vspace*{.05in}
& Lee$$							& 4.181 & 2.822 & 50.299 & 1.861 & 5.818 & 20.843 & 14.304 \\ 
{$\mathbb{P}_\infty$} 
& Lee			& 0.201 & 0.058 & 0.229 & 0.117 & 0.171 & 0.147 & 0.154 \\ 
\hline
\end{tabular}
\vspace*{-.1in}
\end{center}
\small{Note: For Lee's method, only finite sets are used to compute the average diameter and volume, and $\mathbb{P}_\infty$ reports the proportion of excluded infinite sets.} 
\end{table}

Table~\ref{tb:PairwiseqhatC} reports the comparison on pairwise confidence sets for each simulation setting,
where the last column reports the overall averages across all six settings.
While the coverage rates for all confident sets were close to the desirable level, the volume and the diameter of Lee's method were often much larger than our confidence sets. For example, when compared to $\xi(\wh C; \ell_2)$ under $(n,p,\Sigma)=(100,200,\text{EC})$, the average diameter and the average volume of Lee's method were 10 and 55 times larger, respectively.  
To compare the power, we restricted to those pairs $(i,j)$ for which both variables $X_i$ and $X_j$ were in the true support $A_0$, i.e. $i,j\in A_0 \cap A$. As our method built smaller confidence sets, it is not surprising to see that our confidence sets showed a much higher power for all data settings. 

\begin{table}[!t]
\caption{Coverage, power and size of joint confidence sets}
\label{tb:JointqhatC}
\begin{center}
\begin{tabular}{lllllllll}
 \hline 
\multirow{2}{*}{} & {\multirow{2}{*}{Method}}  & \multicolumn{3}{c}{$(n,p) =(100,200)$}   & \multicolumn{3}{c}{$(n,p) =(200,400)$}   & {\multirow{2}{*}{Overall}}\\	 
& & T 	& ED & EC & T 	& ED & EC \\	  \hline  
Coverage & {$\xi_{\mathbf{I}_{q}}(\wh C; \ell_2)$}	 
& 0.900		& 1.000 	& 0.950 	& 0.850		& 0.900		& 0.950 & 0.925 \\ 
& {$\xi_{\mathbf{I}_{q}}(\wh C; \ell_\infty)$}
& 1.000		& 1.000 	& 0.950 	& 1.000		& 1.000		& 1.000 & 0.992\\ \vspace*{.05in}
& Lee
& 0.950		& 0.950 	& 0.950 	& 0.800		& 1.000		& 1.000 & 0.942 \\  
Power & {$\xi_{\mathbf{I}_{q}}(\wh C; \ell_2)$}	 
& 1.000 	& 1.000 	& 0.900 	& 1.000 	& 1.000 	& 1.000 & 0.983\\
& {$\xi_{\mathbf{I}_{q}}(\wh C; \ell_\infty)$}
& 1.000 	& 1.000 	& 0.650 	& 1.000 	& 1.000 	& 1.000 & 0.942 \\ \vspace*{.05in} 
& Lee
& 0.250 	& 0.150 	& 0.050 	& 0.250 	& 0.450 	& 0.050 & 0.200 \\ 
Diameter & {$\xi_{\mathbf{I}_{q}}(\wh C; \ell_2)$}	 
& 1.192		& 	1.093	& 1.509		& 0.891		& 0.785		& 1.324 & 1.132 \\
& {$\xi_{\mathbf{I}_{q}}(\wh C; \ell_\infty)$}
& 2.179		& 	2.080	& 3.101		& 1.576		& 	1.500	& 2.831 & 2.211\\ \vspace*{.05in}
& Lee
& 3.362		& 	3.602	& 8.321		& 2.262		& 	1.188	& 6.471 & 4.201\\ 
Volume$^*$ & {$\xi_{\mathbf{I}_{q}}(\wh C; \ell_2)$}	 
& 0.826 & 0.719 & 0.946 & 0.627 & 0.530 & 0.812 & 0.743 \\
& {$\xi_{\mathbf{I}_{q}}(\wh C; \ell_\infty)$}
& 0.958		& 	0.838	& 1.148		& 0.707		& 0.613		& 1.011 & 0.879\\ \vspace*{.05in}
& Lee
& 1.013 & 1.084 & 3.201 & 0.699 & 0.474 & 1.865 & 1.389\\ 
$\mathbb{P}_{\infty}$ & Lee
& 0.300 	& 0.150		& 0.450		& 0.200		& 0.300		& 0.400 & 0.300 \\ \hline
\end{tabular}
\label{tb:JointCI}
\vspace*{-.1in}
\end{center}
\small{Note: Volume$^*$ is the normalized volume. For Lee's method, only finite sets are used to compute the average diameter and volume, and $\mathbb{P}_\infty$ reports the proportion of excluded datasets due to infinite volumes or infinite diameters.} 
\end{table}

Table~\ref{tb:JointqhatC} reports the results of joint confidence sets for $\nu$, i.e. $H  = \mathbf{I}_q$ in \eqref{eq:JointqhatC}. 
Since an average volume can be highly influenced by a few datasets with large active sets, i.e. large $|A|$, we normalized each volume by the size of $A$ to calculate Volume$^*$ = Volume$^{1/|A|}$ before averaging over datasets. As seen from the table, while the coverage rates of $\xi_{\mathbf{I}_q}(\wh C; \ell_2)$, $\xi_{\mathbf{I}_q}(\wh C; \ell_\infty)$ and Lee's method all stayed around the desirable level, the average diameter and the average volume of Lee's method were larger than ours. In particular, for the equicorrelation designs (EC), the average diameter and normalized volume of Lee's method were, respectively, more than 4 and 2 times larger than those of $\xi_{\mathbf{I}_q}(\wh C; \ell_2)$. When $(n,p,\Sigma) = (200,400,\text{ED})$, we observe that the average volume of Lee's method was smaller than that of ours. This is because we only used datasets for which Lee's method did not produce any infinite intervals when computing the diameters and volumes for their method, which clearly underestimated the average size of their confidence sets. More extreme  differences are seen when comparing the power of a confidence set. While the average power of our method was close to one for most cases, the average power of Lee's method  was only around 0.20, which demonstrates the advantage of our method. The issue of producing infinite confidence sets by Lee's method was even more severe for joint inference, as expected. For $(n,p,\Sigma) = (100,200,\text{EC})$, their method generated infinite intervals for almost half of the datasets, which would be problematic in practical applications.

Overall, our confidence sets were able to achieve the nominal coverage rate with a high power and a small diameter. Our current implementation constructs either a sphere or a hypercube centered at $H\hat{\nu}$ as a confidence set. One may propose alternative ways to build a confidence set of other shapes using the samples of $y^*$ generated by our MCMC algorithm. One option is to approximate the contours of $[HX_A^+y^*\mid \calA(y^*)=A]$ (cf. \eqref{eq:DistJointqhatC}) to build a confidence set, in the similar spirit of a highest posterior density region. We leave this appealing possibility to future work.

\subsection{A case study} \label{sec:caseStudy}

Both our method and Lee's method are based on the truncated Gaussian distribution of $y$ given $\calA(y)=A$, but for some data Lee's intervals turned out to be much wider in the above comparisons. To clarify the key differences between the two methods at a conceptual level, we took a closer look at one dataset from the simulation setting $(n,p,\Sigma)=(100,200,\text{T})$ in Table~\ref{tb:leecomparison1}, for which the lasso support included seven variables, i.e. $|A|=7$. Here, we focus on making inference about $(\nu_2,\nu_5)$, whose true value was $(0.533,0.001)$. The corresponding observed value $(\hat \nu_2,\hat \nu_5)=(\eta_2^\trans y, \eta_5^\trans y)=(0.694,0.148)$, where $\eta_j=(X_A^+)^\trans e_j$.
Our confidence intervals for the two parameters were $\xi_2(\wh C) = [0.266, 1.126]$ and $\xi_5(\wh C) = [-0.263, 0.150]$, respectively, while Lee's intervals $\wh I_2=[0.469,0.918]$ and $\wh I_5=[-11.398,0.413]$. While all four intervals cover the true parameters, $\wh I_5$ is extremely wide compared to $\xi_5(\wh C)$.

Let us walk through our procedure to construct $\xi_j(\wh C)$ in this example. Given an unconditional confidence set $
\wh C$, we first draw $\tdmu^{(i)}$, $i=1,\ldots,K$, uniformly over $\calU(\wh C)$, which are shown as the gray dots in Figure~\ref{fig:Leeillust}(a) after being projected to $\eta_2$ and $\eta_5$. For each $\tdmu^{(i)}$, we simulate a sample of $y^*$ from the conditional distribution $[y^*\mid \calA(y^*)=A]$, i.e. $\dnorm_n(\tdmu^{(i)},\sigma^2\bfI_n)$ truncated to a region $\calD\subset\R^n$, whose projection $(\eta_2^\trans \calD,\eta_5^\trans \calD)$ is illustrated by the solid-line polygon in Figure~\ref{fig:Leeillust}(a). (The exact polygon may differ slightly as we are just plotting an approximate one for easy illustration.) The histograms of the simulated $\eta_j^\trans y^*$ are shown as box-plots in Figure~\ref{fig:Leeillust}(b) and (c) for $j=2,5$. Each box-plot corresponds to the distribution of $\eta_j^\trans y^*$ given one $\tdmu^{(i)}$. We then construct confidence intervals or sets from the aggregation of these samples across all $i=1,\ldots,K$.

\begin{figure}[!t]
	\centering
		\includegraphics[width=1\textwidth]{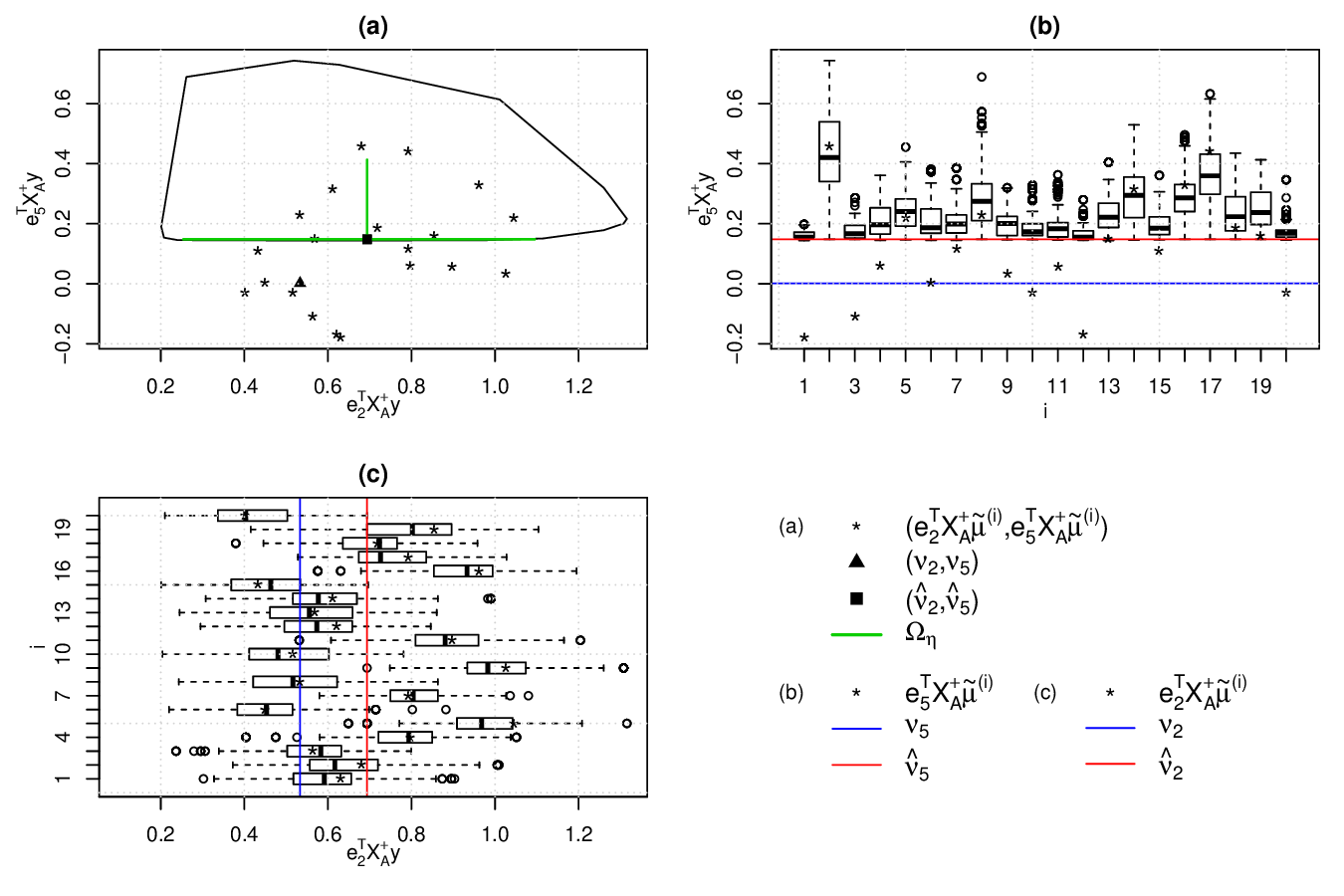}
\vspace*{-.2in}
	\caption{The difference between our and Lee's methods. (a) The feasible regions, (b) box-plots of $e_5^\trans X^+_Ay^*$, and (c) box-plots of $e_2^\trans X^+_Ay^*$.}
	\label{fig:Leeillust}
\end{figure}

To construct the interval $\wh I_j$ for $\nu_j$, \cite{lee:etal:16} decompose $y$ into $\eta_j^\trans y$ and its orthogonal component $z_{-j}\defi (\bfI_n-P_{\eta_j})y$, the residual after projecting to $\eta_j$. Their inference is then based on the conditional distribution 
\begin{align}\label{eq:leeTG}
\eta_j^\trans y\mid \{\calA(y)=A, z_{-j}\}\sim\mathcal{TN}(\nu_j, \sigma^2\|\eta_j\|^2, \Omega_\eta(j)),
\end{align}
where the truncation interval $\Omega_\eta(j)$ depends on the observed value $z^\text{o}_{-j}$ of $z_{-j}$. The green line segments in Figure~\ref{fig:Leeillust}(a) show $\Omega_\eta(j)$, the intersection between $\calD$ and the line $\{y\in\R^n: z_{-j}=z^{\text{o}}_{-j}\}$, which is a bounded one-dimensional interval for each $j$ in this example. They are much more restrictive than the feasible set of our samples $y^*$, projected to these two dimensions, shown as the solid-line polygon. This is the first key difference between the two methods. Then, to define $\wh I_j$, Lee's method finds all possible values of $\nu_j\in\R$ that make the observed statistic $\hat\nu_j$ within $(\alp/2,1-\alp/2)$-quantiles of \eqref{eq:leeTG}. When $\hat\nu_j$ is close to the boundary of $\Omega_\eta(j)$, which is the case for $\hat\nu_5$, their method tends to generate very wide intervals, such as $\wh I_5=[-11.398,0.413]$. Our method gets around this issue with truncated Gaussian inference by considering a smaller range of $\nu_j$ represented by the samples of $\eta_j^\trans \tdmu$, i.e. the star dots ($\star$) in panels (b) and (c). This is the second key difference, and it greatly shortens the constructed intervals, e.g. $\xi_5(\wh C) = [-0.263, 0.150]$. See Figure~\ref{fig:Leeillust}(b) for illustration. On the other hand, when $\hat \nu_j$ is far away from the boundaries of $\Omega_\eta(j)$, as for $j=2$ in this example, Lee's method builds efficient intervals, while ours can be slightly wider due to the additional randomness in $\tdmu$ (Figure~\ref{fig:Leeillust}c).

\section{Discussion}
\label{sec:conc}

We have proposed a new method for post-selection inference, based on estimator augmentation and a conditional MCMC sampler. 
Estimator augmentation is applied to derive a closed-form density for the conditional distribution $[\hbeta_A, S_I \mid \calA(y)=A]$, which is then used as the target distribution in our MCMC sampler. We randomize the estimate of the mean $\mu_0$ by uniform sampling over a confidence set, which incorporates the uncertainty in using a plug-in estimate of $\mu_0$ in our sampling procedure. We have shown with numerical comparisons that our method constructs much shorter confidence intervals than Lee's method \citep{lee:etal:16}, while achieving a comparable coverage rate. Moreover, unlike their method, our method never produces any infinite confidence intervals. With its great flexibility, we further demonstrated that our method can perform joint inference by constructing confidence sets for any set of parameters of interest after lasso selection, which is a unique contribution of this work.

While we have focused on the lasso active set in this work, conditioning on more general events is possible under our framework for post-selection inference. Recall that we parameterize the augmented estimator $(\hbeta,S)$ by the triplet $(\hbeta_\calA,S_\calI,\calA)$. Suppose the selected model is defined by the event $F(\hbeta_\calA,S_\calI,\calA)\in \calE$, where $F$ is a mapping. 
Similar to Corollary~\ref{cor:conden}, we may obtain the density for the conditional distribution
of the augmented estimator given the event $F(\hbeta_\calA,S_\calI,\calA)\in \calE$ based on Theorem~\ref{thm:lasso}. Let $I_\calE(v)$ be the indicator function for $\{v\in \calE\}$.

\begin{corollary}\label{cor:condgeneral}
Under the same assumptions of Theorem~\ref{thm:lasso}, the conditional distribution $[\hbeta_\calA, S_\calI,\calA \mid F(\hbeta_\calA,S_\calI,\calA)\in \calE]$
is given by 
\begin{eqnarray}\label{eq:condgeneral}
\Prob_{\Theta,\calA\mid F\in\calE}(d\theta, A) \propto  f_R\big(V_R^\trans H(\theta, A ; \mu_0, \lambda); \sigma^2 \big)
|\det T(A; \lambda)| I_{\calE}(F(\theta,A))  d\theta,
\end{eqnarray}
where $(\theta,A)=(b_A,s_I,A)$ satisfying the constraints in \eqref{eq:constraint1} and \eqref{eq:constraint2}.
\end{corollary}

Compared to the joint distribution \eqref{eq:joint}, the only difference is the inclusion of the indicator function, which essentially imposes
more constraints on $(\theta,A)$ in addition to \eqref{eq:constraint1} and \eqref{eq:constraint2}. A key step in the development of Monte Carlo algorithms for the above more general conditional distribution is to design efficient proposals that move $(\theta,A)$
in its feasible region satisfying all the imposed constraints. This is a challenging and interesting future direction.
As a concrete example, 
suppose we select variables by thresholding lasso $\hbeta_j$ for $j\in\N_p$,
that is, the selected model is $M = \{j : |\hbeta_j| \geq \tau\}$. 
Then the conditioning event in this example can be written as 
\begin{align*}
\left\{|\hbeta_j| \geq \tau, \forall\, j \in M\right\} \bigcap  \left\{|\hbeta_j| < \tau, \forall\, j\notin M\right\}.
\end{align*}
Note that the active set $\calA$ is no longer fixed on this event, although it must satisfy $\calA\supset M$. As a result, the target distribution \eqref{eq:condgeneral} is a joint distribution for $\Theta$ and $\calA$, which will incur additional computational cost to our MH sampler. In particular, the dimension of $\Theta$ will change when the size of $\calA$ changes.
A normal distribution centered at $b_j^{(t)}$ and truncated to $(-\infty, -\tau]\cap[\tau, \infty)$ can be used as a proposal for $b_j^\dagger$, $j\in M$ at the $t$-th iteration. For ${j \notin M}$, we need to consider some of the active $b_j^{(t)}\in(-\tau,\tau)$ turning into zero and vice versa, which would change the active set $\calA$.

Post-selection inference on data with group structure using estimator augmentation for group lasso can be another interesting future topic. Due to the complex sample space of the augmented group lasso estimator \citep{zhou:min:17b}, developing an MCMC sampler is more complicated than the one for lasso. However, this potential generalization will be important for applications with pre-grouped variables, categorical variables, or highly correlated predictors.

\section{Proofs} \label{sec:proof}  
\begin{proof}[Proof of Proposition~\ref{lm:converage}]
From \eqref{eq:defupmax} and \eqref{eq:deflowmax}, we have 
$\xi_j(\mu_0, \alpha/2)\subset\xi_j^*(\wh C)$ if $\mu_0\in \wh C$. Then,
\begin{eqnarray*}
\Prob\left\{\nu_j\in  \xi_j^*(\wh C) \bigg| \calA(y) = A \right\} 
&\geq& \Prob\left\{\nu_j\in  \xi_j^*(\wh C) \bigg| \calA(y) = A, \mu_0 \in \wh C \right\} \Prob\big(\mu_0 \in \wh C \mid \calA(y) = A\big)\\ 
&\geq& \Prob\left\{\nu_j\in  \xi_j(\mu_0, \alpha/2) \bigg| \calA(y) = A, \mu_0 \in \wh C\right\} \Prob\big(\mu_0 \in \wh C \mid \calA(y) = A\big).
\end{eqnarray*}
Due to the independence between $\wh C$ and $y$,
\begin{align*}
\Prob\left\{\nu_j\in  \xi_j(\mu_0, \alpha/2) \bigg| \calA(y) = A, \mu_0 \in \wh C\right\} = 
\Prob\left\{\nu_j\in  \xi_j(\mu_0, \alpha/2) \bigg| \calA(y) = A\right\} = 1-\alpha/2,
\end{align*}
and $\Prob\big(\mu_0 \in \wh C \mid \calA(y) = A\big) = \Prob(\mu_0 \in \wh C) = 1 - \alpha / 2,$ which imply \eqref{eq:conservativepr}.
\end{proof}

\begin{proof}[Proof of Theorem~\ref{thm:condsample}]
Under the assumptions of Theorem~\ref{thm:lasso}, for every $y\in\R^n$ there is a unique $(\hbeta,S)$ (Lemma 1 in \cite{zhou:14}). Therefore, the KKT condition~\eqref{eq:lassoKKT} establishes a bijection between $y$
and the augmented estimator $(\hbeta,S)$. Consequently, $y$ can be uniquely represented by
\begin{align}
 y&= (X^\trans)^+(X^{\trans} X \hbeta + n\lambda W S) \nonumber\\
 &=X_A \hbeta_A + n\lambda (X^\trans)^+ \big\{ W_A \sgn(\hbeta_A) + W_I S_I \big\}, \label{eq:yrepr}
\end{align}
where $(X^\trans)^+ = (XX^\trans)^{-1}X$ because $X$ has full row rank.
Since $|A|\leq n$ (Remark \ref{rm:general}) and every $n$ columns of $X$ are linearly independent, we have 
$X_A^+=(X_A^\trans X_A)^{-1}X_A^\trans$.
From \eqref{eq:lassoKKT}, 
\begin{equation*}
X_A^\trans y = X_A^\trans X_A \hbeta_A + n\lambda W_{AA}\sgn(\hbeta_A).\\ 
\end{equation*}
Multiplying both sides by $(X_A^\trans X_A)^{-1}$, we get
\begin{equation}\label{eq:XAyrepr}
X_A^+ y = \hbeta_A + n\lambda (X_A^\trans X_A)^{-1} W_{AA}\sgn(\hbeta_A).\\ 
\end{equation}
Then the conclusions~\eqref{eq:represent} and \eqref{eq:XAy} follow from \eqref{eq:yrepr} and \eqref{eq:XAyrepr}
due to the assumption that $[\hbeta^*_A,S^*_I]=[\hbeta_A,S_I\mid \calA =A]$.
\end{proof}

\bibliographystyle{asa}
\bibliography{SparseInference} 

\end{document}

%% file: Packages.tex
\usepackage{geometry}                
\usepackage{graphicx}
\usepackage{amsmath,amssymb, amsthm,epsfig,bm}
\usepackage{setspace}
\usepackage{epstopdf}
\usepackage{psfrag}
\usepackage{color,soul}
\usepackage{verbatim}
\usepackage{algorithm,algorithmic}
\usepackage{multirow}
\usepackage{natbib}
\usepackage{xr}
\usepackage{mathrsfs}
\usepackage{cases}
\usepackage{enumitem}


%% file: MyPageStyle.tex
\setlist{noitemsep, topsep=0pt}
\bibpunct{(}{)}{;}{a}{}{,} 
\usepackage[english]{babel}
\usepackage[utf8]{inputenc}
\usepackage{fancyhdr}
 


%% file: NewCommands.tex
\newcommand{\bfI}{\mathbf{I}}

\newcommand{\calA}{\mathcal{A}}

\newcommand{\calD}{\mathcal{D}}
\newcommand{\calI}{\mathcal{I}}

\newcommand{\calN}{\mathcal{N}}

\newcommand{\calT}{\mathcal{T}}
\newcommand{\calU}{\mathcal{U}}
\newcommand{\calV}{\mathcal{V}}

\newcommand{\R}{\mathbb{R}}
\newcommand{\N}{\mathbb{N}}

\newcommand{\Prob}{\mathbb{P}}

\newcommand{\alp}{\alpha}

\newcommand{\hbeta}{\hat{\beta}}

\newcommand{\veps}{\varepsilon}

\newcommand{\hmu}{\hat{\mu}}
\newcommand{\tdmu}{\tilde{\mu}}

\newcommand{\Omg}{\Omega}

\newcommand{\calE}{\mathcal{E}}

\newcommand{\defi}{\mathop{:=}} 
  
\newcommand{\dnorm}{\mathcal{N}}

\newcommand{\trans}{\mathsf{T}}

\newcommand{\wh}{\widehat}

\DeclareMathOperator*{\argmin}{argmin}

\DeclareMathOperator{\diag}{diag}
\DeclareMathOperator{\sgn}{sgn}
\DeclareMathOperator{\row}{row}

\DeclareMathOperator{\nul}{null}
\DeclareMathOperator{\rank}{rank}
\DeclareMathOperator{\supp}{supp}

\RequirePackage{tikz}

%% file: PostSelectionInference.bbl
\begin{thebibliography}{21}
\newcommand{\enquote}[1]{``#1''}
\expandafter\ifx\csname natexlab\endcsname\relax\def\natexlab#1{#1}\fi

\bibitem[{Bachoc et~al.(2020)Bachoc, Preinerstorfer, and
  Steinberger}]{Bachoc20}
Bachoc, F., Preinerstorfer, D., and Steinberger, L. (2020), \enquote{Uniformly
  valid confidence intervels post-model-selection,} \textit{Annals of
  Statistics}, 48, 440--463.

\bibitem[{Berk et~al.(2013)Berk, Brown, Buja, Zhang, and Zhao}]{berk:etal:13}
Berk, R., Brown, L., Buja, A., Zhang, K., and Zhao, L. (2013), \enquote{Valid
  post-selection inference,} \textit{Ann. Statist.}, 41, 802--837.

\bibitem[{Ewald and Schneider(2018)}]{ewald18}
Ewald, K. and Schneider, U. (2018), \enquote{Uniformly valid confidence sets
  based on the Lasso,} \textit{Electron. J. Statist.}, 12, 1358--1387.

\bibitem[{Ewald and Schneider(2020)}]{Ewald20}
--- (2020), \enquote{On the distribution, model selection properties and
  uniqueness of the lasso estimator in low and high dimensions,}
  \textit{Electronic Journal of Statistics}, 14, 944--969.

\bibitem[{Kivaranovic and Leeb(2018)}]{kiva:leeb:18}
Kivaranovic, D. and Leeb, H. (2018), \enquote{Expected length of
  post-model-selection confidence intervals conditional on polyhedral
  constraints,} \textit{arXiv:1803.01665}.

\bibitem[{Lee et~al.(2016)Lee, Sun, Sun, and Taylor}]{lee:etal:16}
Lee, J.~D., Sun, D.~L., Sun, Y., and Taylor, J.~E. (2016), \enquote{Exact
  post-selection inference, with application to the lasso,} \textit{Ann.
  Statist.}, 44, 907--927.

\bibitem[{Leeb and P{\"o}tscher(2006)}]{leeb:pots:06}
Leeb, H. and P{\"o}tscher, B.~M. (2006), \enquote{Can one estimate the
  conditional distribution of post-model-selection estimators?} \textit{Ann.
  Statist.}, 34, 2554--2591.

\bibitem[{Liu et~al.(2018)Liu, Markovic, and Tibshirani}]{liu:etal:18}
Liu, K., Markovic, J., and Tibshirani, R. (2018), \enquote{More powerful
  post-selection inference, with application to the Lasso,}
  \textit{arXiv:1801.09037}.

\bibitem[{Lockhart et~al.(2014)Lockhart, Taylor, Tibshirani, and
  Tibshirani}]{lock:etal:14}
Lockhart, R., Taylor, J., Tibshirani, R.~J., and Tibshirani, R. (2014),
  \enquote{A significance test for the lasso,} \textit{Ann. Statist.}, 42,
  413--468.

\bibitem[{Nickl and van~de Geer(2013)}]{nick:vand:13}
Nickl, R. and van~de Geer, S. (2013), \enquote{Confidence sets in sparse
  regression,} \textit{Ann. Statist.}, 41, 2852--2876.

\bibitem[{P{\"o}tscher(1991)}]{pots:91}
P{\"o}tscher, B.~M. (1991), \enquote{Effects of model selection on inference,}
  \textit{Econom. Theory}, 7, 163--185.

\bibitem[{Taylor and Tibshirani(2018)}]{tayl:tibs:18}
Taylor, J. and Tibshirani, R. (2018), \enquote{Post-selection inference for
  {$\ell_1$}-penalized likelihood models,} \textit{Canad. J. Statist.}, 46,
  41--61.

\bibitem[{Tian and Taylor(2017)}]{tian:tayl:17}
Tian, X. and Taylor, J. (2017), \enquote{Asymptotics of selective inference,}
  \textit{Scand. J. Stat.}, 44, 480--499.

\bibitem[{Tian and Taylor(2018)}]{tiantaylor18}
--- (2018), \enquote{Selective inference with a randomized response,}
  \textit{Annals of Statistics}, 46, 679--710.

\bibitem[{Tibshirani(1996)}]{tibs:96}
Tibshirani, R. (1996), \enquote{{Regression selection and shrinkage via the
  lasso},} \textit{Journal of the Royal Statistical Society: Series B
  (Methodological)}, 58, 267--288.

\bibitem[{Tibshirani(2013)}]{tibs:13}
Tibshirani, R.~J. (2013), \enquote{The lasso problem and uniqueness,}
  \textit{Electron. J. Stat.}, 7, 1456--1490.

\bibitem[{Tibshirani et~al.(2018)Tibshirani, Rinaldo, Tibshirani, and
  Wasserman}]{tibs:etal:18}
Tibshirani, R.~J., Rinaldo, A., Tibshirani, R., and Wasserman, L. (2018),
  \enquote{Uniform asymptotic inference and the bootstrap after model
  selection,} \textit{Ann. Statist.}, 46, 1255--1287.

\bibitem[{Tibshirani et~al.(2016)Tibshirani, Taylor, Lockhart, and
  Tibshirani}]{tibs:etal:16}
Tibshirani, R.~J., Taylor, J., Lockhart, R., and Tibshirani, R. (2016),
  \enquote{Exact post-selection inference for sequential regression
  procedures,} \textit{J. Amer. Statist. Assoc.}, 111, 600--620.

\bibitem[{Zhou et~al.(2019)Zhou, Li, and Zhou}]{zhou:etal:19}
Zhou, K., Li, K.-C., and Zhou, Q. (2019), \enquote{Honest confidence sets for
  high-dimensional regression by projection and shrinkage,}
  \textit{arXiv:1902.00535}.

\bibitem[{Zhou(2014)}]{zhou:14}
Zhou, Q. (2014), \enquote{{Monte Carlo simulation for lasso-type problems by
  estimator augmentation},} \textit{J. Amer. Statist. Assoc.}, 109, 1495--1516.

\bibitem[{Zhou and Min(2017)}]{zhou:min:17b}
Zhou, Q. and Min, S. (2017), \enquote{Estimator augmentation with applications
  in high-dimensional group inference,} \textit{Electron. J. Statist.}, 11,
  3039--3080.

\end{thebibliography}
